\documentclass[10pt,twocolumn,twoside]{IEEEtran}

\usepackage[sort,compress]{cite}
\usepackage[final]{graphicx}
\usepackage{epstopdf}

\usepackage{hyperref}

\usepackage[cmex10]{mathtools}
\usepackage{dsfont,amsfonts,mathrsfs}
\usepackage[cmex10]{amsmath}
\usepackage{amsthm} 

\theoremstyle{plain}
\newtheorem {problem}{Problem}

\theoremstyle{plain}

\theoremstyle{plain}
\newtheorem {theorem}{Theorem}

\theoremstyle{plain}

\newtheorem {proposition}{Proposition}

\theoremstyle{definition}
\newtheorem {definition}{Definition}
\newtheorem{property}{Property}

\theoremstyle{remark}
\newtheorem{remark}{Remark}

\theoremstyle{plain}
\newtheorem {lemma}{Lemma}

\newtheorem{lemsec}{Lemma}[section]

\theoremstyle{plain}

\theoremstyle{plain}
\newtheorem* {assumption}{\edit{Assumption (A)}}

\usepackage[T1]{fontenc}

\newcommand\IND{\mathds{1}}
\newcommand\PR {\mathds{P}}
\newcommand\EXP{\mathds{E}}

\newcommand\reals{\mathds{R}}

\newcommand\ALPHABET{\mathcal}

\newcommand{\tBELL}[3]{[ \tilde {\mathscr B}_{#1} {#2} ]({#3})}

\usepackage{xcolor}
\newcommand\edit{\textcolor{blue}}
\let\edit\relax

\begin{document}

\title{Information-Theoretic Privacy for Smart Metering Systems with a Rechargeable Battery}

\author{Simon Li, Ashish Khisti, and Aditya Mahajan
  \thanks{Simon Li  was with the University of Toronto when this work was done. Ashish Khisti (akhisti@ece.utoronto.ca) is with the University of Toronto, Toronto, ON, Canada. Aditya Mahajan (aditya.mahajan@mcgill.ca) is with McGill University, Montreal, QC, Canada. }
  \thanks{Preliminary version of this work was presented at the Signal Processing for Advanced Wireless Communications (SPAWC), Jul.~2015, Stockholm, Sweden,  International Zurich Seminar (IZS), Mar.~2-4, 2016, Zurich, Switzerland and the American Control Conference (ACC), July 6-8, 2016, Boston, MA, USA.}}
\maketitle

% ===================================================================================
\begin{abstract} 
Smart-metering systems report electricity usage of a user to the utility provider on almost real-time basis.
This could leak private information about the user to the utility provider.
 In this work we investigate the use of a rechargeable battery in order
to provide privacy to the user. %from the utility provider.

We assume that the user load sequence  is a first-order Markov process, the
battery satisfies ideal charge conservation, and that privacy is measured
using normalized mutual information (leakage rate) between the user load and
the battery output. We study the optimal battery charging policy that minimizes
the leakage rate among the class of battery policies that satisfy causality and
charge conservation. We propose a series 
reductions on the original problem and
ultimately recast it as a Markov Decision Process (MDP) that can be solved
using a dynamic program.

In the special case of i.i.d.\ demand, we explicitly
characterize the optimal policy and show that the associated leakage rate
can be expressed as a single-letter mutual information expression. In this case
we show that the optimal charging policy admits an intuitive interpretation of preserving
a certain invariance property of the state. Interestingly an alternative proof of optimality
can be provided that does not rely on the MDP approach, but is based on purely information
theoretic reductions.
\end{abstract}

\IEEEpeerreviewmaketitle

% ===================================================================================
\section{Introduction}

Smart meters are a critical part of modern power distribution systems because
they provide fine-grained power consumption measurements to utility providers.
These fine-grained measurements improve the efficiency of the power grid by
enabling services such as time-of-use pricing and demand
response~\cite{Ipakchi09}. However, this promise of improved efficiency is
accompanied by a risk of privacy loss. It is possible for the utility
provider---or an eavesdropper---to infer private information including load
taxonomy from the fine-grained measurements provided by smart
meters~\cite{infer1,infer2,infer3}. Such private information could be
exploited by third parties for the purpose of targeted advertisement or
surveillance. Traditional techniques in which an intermediary anonymizes the
data~\cite{anon} are also prone {\edit{to privacy}} loss to an eavesdropper. One possible
solution is to partially obscure the load profile by using a rechargeable
battery~\cite{smart-meter-1}. As the cost of rechargeable batteries decreases
(for example, due to proliferation of electric vehicles), using them for improving privacy is becoming economically viable.

In a smart metering system with a rechargeable battery, the energy consumed
from the power grid may either be less than the user's demand---the rest being
supplied by the battery; or may be more than the user's demand---the excess
being stored in the battery. A rechargeable battery provides privacy because
the power consumed from the grid (rather than the user's demand) gets reported
to the electricity utility (and potentially observed by an eavesdropper). In
this paper, we focus on the mutual information between the user's demand and
consumption (i.e., the information leakage) as the privacy metric. Mutual Information is a widely
used metric in the literature on information theoretic security, as it is often analytically tractable
and provides a fundamental bound on the probability of detecting the true load sequence from the
observation~\cite{Csiszar}. Our
objective is to identify a battery management policy (which determine how much
energy to store or discharge from the battery) to minimize the information
leakage rate.

We briefly review the relevant literature. The use of a rechargeable battery for providing
user privacy has been studied in several recent works, e.g.,~\cite{smart-meter-1,varodayan,Tan12,giaconi,pappas}.  Most
of the existing literature has focused on evaluating the information leakage
rate of specific battery management policies. These include the
``best-effort'' policy~\cite{smart-meter-1}, which tries to maintain a
constant consumption level, whenever possible; and battery conditioned
stochastic charging policies~\cite{varodayan}, in which the conditional
distribution on the current consumption depends only on the current battery
state (or on the current battery state and the current demand).
In~\cite{smart-meter-1}, the information leakage rate was estimated using
Monte-Carlo simulations; in~\cite{varodayan}, it was calculated using the BCJR
algorithm~\cite{bcjr}. The methodology of~\cite{varodayan} was extended
by~\cite{Tan12} to include models with energy harvesting and allowing for a certain
amount of energy waste. Bounds on the
performance of the best-effort policy and hide-and-store policy for models
with energy harvesting and infinite battery capacity were obtained
in~\cite{giaconi}. The performance of the best effort algorithm for an
individual privacy metric was considered in~\cite{pappas}.
None of these papers address the question of choosing the
optimal battery management policy.

Rate-distortion type approaches have also been used to study privacy-utility
trade-off~\cite{sankar,Rajagopalan,gunduz}. These models allow the user to report
a distorted version of the load to the utility provider, subject to a certain average distortion constraint.
Our setup differs from these works as we impose a constraint on the
\emph{instantaneous} energy stored in the battery due to its limited capacity. 
Both our techniques and the qualitative nature of the results are different from these papers.

Our contributions are two-fold. First, when the demand is Markov, we show that
the minimum information leakage rate and optimal battery management policies
can be obtained by solving an appropriate dynamic program. These results are
similar in spirit to the dynamic programs obtained to compute capacity of
channels with memory~\cite{Tatikonda09,goldsmith,permuter}; however, the
specific details are different due to the constraint on the battery state.
Second, when the demand is i.i.d., we obtain a single letter characterization
of the minimum information leakage rate; this expression also gives the
optimal battery management policy. We prove the single letter expression in
two steps. On the achievability side we propose a class of  policies with  a specific
structure that enables a considerable simplification of the leakage-rate expression.
We find a policy that minimizes the leakage-rate within this restricted class.
On the converse side, we obtain lower bounds on the minimal leakage rate and show that these
lower bound match the performance of the best structured policy. We provide
two proofs. One is based on the dynamic program and the other
is based purely on information theoretic arguments.

After the present work was completed, we became aware of~\cite{parv}, where a
similar dynamic programming framework is presented for the infinite horizon
case. However, no explicit solutions of the dynamic program are derived
in~\cite{parv}. To the best of our knowledge, the present paper is
the first work that provides an explicit characterization of the optimal
leakage rate and the associated policy for i.i.d.\ demand.  

\subsection{Notation}

Random variables are denoted by uppercase letters ($X$, $Y$, etc.), their
realization by corresponding lowercase letters ($x$, $y$, etc.), and their
state space by corresponding script letters ($\ALPHABET X$, $\ALPHABET Y$,
etc.). $\ALPHABET P_X$ denotes the space of probability distributions on
$\ALPHABET X$; $\ALPHABET P_{X|Y}$ denotes the space of \edit{condition
distributions} from $\ALPHABET Y$ to $\ALPHABET X$. $X_a^b$ is a short hand for $(X_a,
X_{a+1}, \dots, X_b)$ and $X^b = X_{1}^b$. For a set $\ALPHABET A$,
$\IND_{\ALPHABET A}(x)$ denotes the indicator function of the set that equals
$1$ if $x \in \ALPHABET A$ and zero otherwise. If $\ALPHABET A$ is a singleton
set $\{a\}$, we use $\IND_{a}(x)$ instead of $\IND_{\{a\}}(x)$.

Given random variables $(X,Y)$ with joint distribution $P_{X,Y}(x,y) = P_X(x)
q(y|x)$, $H(X)$ and $H(P_X)$ denote the entropy of $X$, $H(Y|X)$ and
$H(q|P_X)$ denote conditional entropy of $Y$ given $X$ and $I(X;Y)$ and 
$I(q; P_X)$ denote the mutual information between $X$ and $Y$. 

%------------------------------------------------------------------------
\section{Problem Formulation and Main Results}

\subsection{Model and problem formulation}
\label{subsec:model}
Consider a smart metering system as shown in Fig.~\ref{fig:system}. At each
time, the energy consumed from the power grid must equal the user's demand
plus the additional energy that is either stored in or drawn from the battery. Let $\{X_t\}_{t \ge 1}$, $X_t \in
\ALPHABET X$, denote the user's demand; $\{Y_t\}_{t \ge 1}$, $Y_t \in
\ALPHABET Y$, denote the energy drawn from the grid; and $\{S_t\}_{t \ge 1}$,
$S_t \in \ALPHABET S$, denote the energy stored in the battery. All alphabets
are finite. For convenience, we assume $\ALPHABET X \coloneqq \{0, 1, \dots,
m_x\}$, $\ALPHABET Y \coloneqq \{0, 1, \dots, m_y\}$, and $\ALPHABET S =
\{0,1, \dots, m_s\}$. \edit{Here $m_s$ corresponds to the size of the
battery.} We note that such a restriction is for simplicity of presentation; the results generalize even when $\ALPHABET X$ and
$\ALPHABET Y$ are not necessarily contiguous intervals or integer valued.  To guarantee that
user's demand is always satisfied, we assume $m_x \le m_y$ or that
${\ALPHABET X} \subseteq {\ALPHABET Y}$ holds more generally.

The demand $\{X_t\}_{t \ge 1}$ is a first-order time-homogeneous Markov
chain\footnote{In practice, the energy demand is periodic rather than time
  homogeneous. We are assuming that the total demand may be split into a
  periodic predictable component and a time-homogeneous stochastic component.
  In this paper, we ignore the predictable component because it does not
affect privacy.} with transition probability $Q$. We assume that 
$Q$ is irreducible and aperiodic. The initial state $X_1$ is
distributed according to probability mass function $P_{X_1}$. The initial
charge $S_1$ of the battery is independent of $\{X_t\}_{t \ge 1}$ and
distributed according to probability mass function $P_{S_1}$. 

The battery is assumed to be ideal and has no conversion losses or other
inefficiencies. Therefore, the following conservation equation must be
satisfied at all times:
\begin{equation}
  \label{eq:conservation}
  S_{t+1} = S_t + Y_t - X_t.
\end{equation}

\begin{figure}[t]
  \centering
  \includegraphics[width=\linewidth]{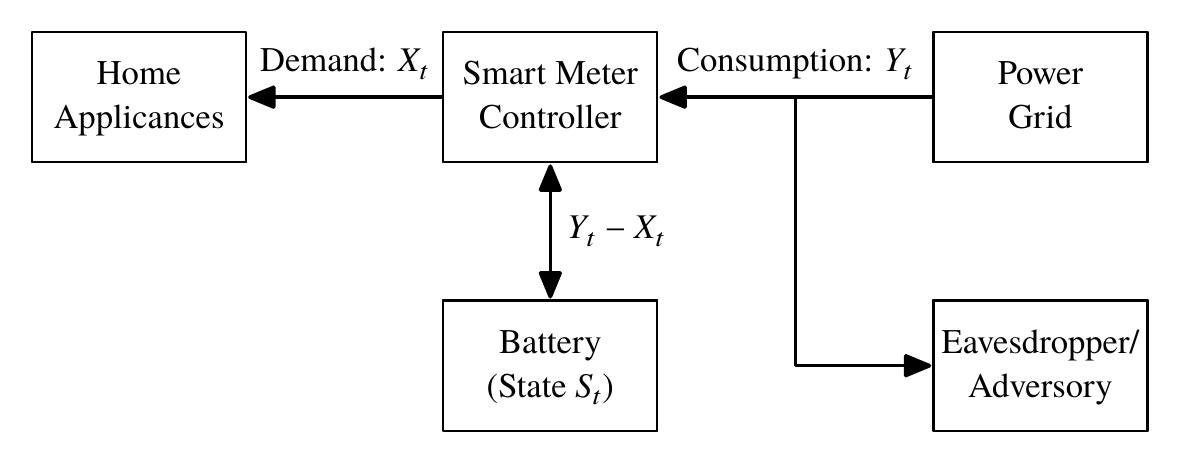}
  \caption{A smart metering system.}
  \label{fig:system}
\end{figure}

Given the history of demand, battery charge, and consumption, a randomized
\emph{battery charging policy} $\mathbf q = (q_1, q_2, \dots)$ determines the
energy consumed from the grid. In particular, given the histories $(x^{t}$,
$s^{t}, y^{t-1})$ of demand, battery charge, and consumption at time~$t$, the
probability that current consumption $Y_t$ equals $y$ is $q_t(y \mid x^{t},
s^{t}, y^{t-1})$. For a randomized charging policy to be feasible, it must
satisfy the conservation equation~\eqref{eq:conservation}. So, given the
current power demand and battery charge $(x_t,s_t)$, the feasible values of
grid consumption are defined by
\begin{align} 
  \ALPHABET Y_\circ(s_t-x_t) = \{y \in \ALPHABET Y : s_t-x_t + y \in \ALPHABET S \}. \label{eq:yo-def}
\end{align}
Thus, we require that
\begin{equation*} 
  q_t(\ALPHABET Y_\circ(s_{t}-x_t) \mid x^{t}, s^{t}, y^{t-1}) 
  := \smashoperator{\sum_{y \in \ALPHABET Y_\circ(s_t - x_{t})}}
  q_t(y \mid x^{t}, s^{t}, y^{t-1}) = 1.
\end{equation*} 
The set of all such feasible policies is denoted by $\ALPHABET
Q_A$\footnote{With a slight abuse of notation, we use $\mathcal Q_A$ to denote
the battery policy for both the infinite and finite-horizon problems}. Note that while
the charging policy $q_t(\cdot)$ can be a function of the entire history, the support 
of $q_t(\cdot)$ only depends on the present value of $x_t$ and $s_t$ through the difference $s_t-x_t$.
This is emphasized in the definition in~\eqref{eq:yo-def}.

The quality of a charging policy depends on the amount of information leaked
under that policy. There are different notions of privacy; in this paper, we
use mutual information as a measure of privacy. Intuitively speaking, given
random variables $(Y,Z)$, the mutual information $I(Y;Z)$ measures the
decrease in the uncertainty about $Y$ given by $Z$ (or vice-versa). Therefore,
given a policy $\mathbf{q}$, the information about $(X^T,S_1)$ leaked to the
utility provider or eavesdropper is captured by $I^\mathbf{q}(X^T, S_1; Y^T)$, where the
mutual information is evaluated according to the joint probability
distribution on $(X^{T}, S^{T}, Y^{T})$ induced by the distribution $\mathbf{q}$
as follows:
\begin{multline*} 
  \PR^{\mathbf{q}}( S^{T} = s^{T}, X^{T} = x^{T}, Y^{T} = y^{T}) \\
  = P_{S_1}(s_1) P_{X_1}(x_1) q_1(y_1 \mid x_1, s_1)  
  \prod_{t=2}^{T} \bigg[ \IND_{s_{t}}\{s_{t-1} - x_{t-1} + y_{t-1} \} \\
    \times Q(x_{t}| x_{t-1}) q_t(y_{t} \mid x^{t}, s^{t}, y^{t-1}) 
  \bigg].
\end{multline*}

We use information leakage \emph{rate} as a measure of the quality of a
charging policy. For a finite planning horizon, the information leakage rate of
a policy $\mathbf{q} = (q_1, \dots, q_T) \in \ALPHABET Q_A$ is given by
\begin{equation}
  \label{eq:leakage-finite}
  L_T(\mathbf{q}) := \frac 1T I^\mathbf{q}(X^T, S_1; Y^T),
\end{equation}
while for an infinite horizon, the worst-case information leakage rate of a
policy $\mathbf{q} = (q_1, q_2, \dots) \in \ALPHABET Q_A$ is given by
\begin{equation}
  \label{eq:leakage-infinite}
  L_\infty(\mathbf{q}) := \limsup_{T \to \infty} L_T(\mathbf{q}).
\end{equation}

We are interested in the following optimization problems:
\begin{problem} \label{prob:original}
  Given the alphabet $\ALPHABET X$ of the demand, the initial distribution
  $P_{X_1}$ and the transistion matrix $Q$ of the demand process, the alphabet
  $\ALPHABET S$ of the battery, the initial distribution $P_{S_1}$ of the
  battery state, and the alphabet $\ALPHABET Y$ of the consumption:
  \begin{enumerate}
    \item For a finite planning horizon $T$, find a battery charging policy
      $\mathbf{q} = (q_1, \dots, q_T) \in \ALPHABET Q_A$ that minimizes the leakage rate
      $L_T(\mathbf{q})$ given by~\eqref{eq:leakage-finite}.
    \item For an infinite planning horizon, find a battery charging policy $\mathbf{q}
      = (q_1, q_2, \dots) \in \ALPHABET Q_A$ that minimizes the leakage rate
      $L_\infty(\mathbf{q})$ given by~\eqref{eq:leakage-infinite}.
  \end{enumerate}
\end{problem}

The above optimization problem is difficult because we have to optimize a
multi-letter mutual information expression over the class of history
dependent probability distributions. In the spirit of results for feedback
capacity of channels with memory~\cite{Tatikonda09,goldsmith,permuter}, we
show that the above optimization problem can be reformulated as a Markov
decision process where the state and action spaces are conditional
probability distributions. Thus, the optimal policy and the optimal leakage
rate can be computed by solving an appropriate dynamic program. We then
provide an explicit solution of the dynamic program for the case of i.i.d.\
demand. 

\begin{remark}
\label{rem:det}
\edit{We note that the class of policies we consider in $\ALPHABET Q_A$ are randomized policies i.e., the output at any given time is governed by 
the conditional distribution $q_t(y_{t} \mid x^{t}, s^{t}, y^{t-1})$. The class of {\em deterministic policies} where $y_t= h_t(x^{t}, s^{t}, y^{t-1})$ is  a deterministic
function of the past inputs, state and outputs is a special case of randomized
policies where $q_t(\cdot)$ is an atomic distribution. As will be apparent
from out results, deterministic policies do not suffice to minimize the
leakage rate and hence we focus on the class of randomized policies. }
\end{remark}
\subsection{Example: Binary Model}
\label{subsec:bin-model}
We illustrate the special case when $\ALPHABET X = \ALPHABET Y = \ALPHABET S = \{0,1\}$ in Fig.~\ref{fig:Binary}.
The input, output,  as well as the state, are all binary valued. When the battery is in state $s_t=0$, there are three possible transitions.
If the input $x_t=1$ then we must select $y_t=1$ and the state changes to $s_{t+1}=0$. If instead $x_t=0$, then there are two possibilities. We can select
$y_t=0$ and have $s_{t+1}=0$ or we can select $y_t=1$ and have $s_{t+1}=1$. In a similar fashion there are three possible transitions from the state $s_t=1$
as shown in Fig.~\ref{fig:Binary}. We will assume that the demand (input)
sequence  is sampled i.i.d.\ from an equiprobable distribution, \edit{i.e.,
$P_X(0) = P_X(1) = \frac 12$.} 
\begin{figure}
  \centering
  \includegraphics{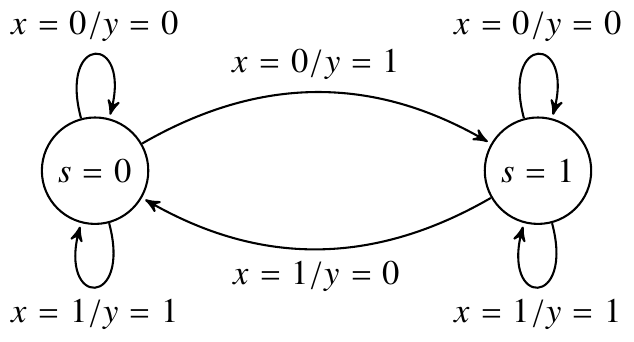}
  \caption{Binary System model. The battery can be either in $s=0$ or $s=1$.
  The set of feasible transitions from each state are shown in the figure.}
\label{fig:Binary}
\end{figure}

Consider a simple policy that sets $y_t=x_t$ and ignores the battery state. It is clear that such a policy will lead to maximum leakage $L_T=1$.
Another feasible policy  is to set $y_t = \edit{1 - {s}_t}$. Thus whenever $s_t=0,$ we will set $y_t=1$ regardless of the value of $x_t,$ and likewise $s_t=1$ will result in $y_t=0$.
It turns out that the leakage rate for this policy also approaches $1$. To see this note that the eavesdropper having access to $y^T$ also in turn knows $s^T$. Using the battery update equation~\eqref{eq:conservation} the sequence $x^{T-1}$ is thus revealed to the eavesdropper, resulting in a leakage rate of at least $1-1/T$. 

In reference~\cite{varodayan}, a probabilistic battery charging policy is introduced that only depends on the current state and input i.e., $q_t(y_t |x^t, s^t) \stackrel{\Delta}{=} q_t(y_t | x_t, s_t)$. Furthermore the policy makes equiprobable decisions between the feasible transitions i.e., 
\begin{equation}\label{eq:varodayan}
  q_t(y_t =0 | x_t, s_t) = q(y_t =1 | x_t, s_t)  = 1/2, \quad \text{\edit{if} } x_t = s_t
\end{equation}
and $q_t(y_t | x_t ,s_t) = \IND_{x_t}(y_t)$ otherwise. The leakage rage for this policy was numerically evaluated in~\cite{varodayan} using the BCJR algorithm and it was shown numerically that $L_\infty = 0.5$. Such numerical techniques seem necessary in general even for the class of memoryless policies and i.i.d.\ inputs, as the presence of the battery adds memory into the system.

As a consequence of our main result it follows that the above policy  admits a
single-letter expression for the leakage rate\footnote{The random variable
$S^*$ is an equiprobable binary valued random variable, independent of $X$.
See Sec.~\ref{sec:revisit}.} $L_\infty = I(S^*-X; X)$, thus  circumventing the need for numerical techniques. Furthermore it also follows that this leakage rate is indeed the minimum possible one among the class of all feasible policies. Thus it is not necessary for the battery system to use more complex policies that take into account the entire history. We note that a similar result was shown in~\cite{spawc} for the case of finite horizon policies. However the proof in~\cite{spawc} is specific to the binary model. In the present paper we provide a complete single-letter solution to the case of general i.i.d.\ demand, and a dynamic programming method for the case of first-order Markovian demands, as discussed next.

\subsection{Main results for Markovian demand}
We identify two \edit{structural} simplifications for the battery charging policies. 
First, we show (see Proposition~\ref{prop:A-B} in Section~\ref{sec:A-B})
that there is no loss of optimality in restricting attention to charging
policies of the form
\begin{equation} \label{eq:qb}
  q_t(y_t | x_t, s_t, y^{t-1}).
\end{equation}
The intuition is that under such a policy, observing $y^t$ gives partial
information only about $(x_t, s_t)$ rather than about the whole history
$(x^t, s^t)$. 

Next, we identify a sufficient statistic for $y^{t-1}$ \edit{when} the charging
policies \edit{are} of the form~\eqref{eq:qb}. For that matter, 
given a policy $\mathbf q$ and any realization $y^{t-1}$ of $Y^{t-1}$, define
the belief state $\pi_t \in \ALPHABET P_{X,S}$ as follows: \edit{for all $x
\in \ALPHABET X$, $s \in \ALPHABET S$},
\begin{equation}\label{eq:pi-def}
  \pi_t(x,s) = \PR^{\mathbf q}(X_t = x, S_t = s | Y^{t-1} = y^{t-1}).
\end{equation}
% Note that $\pi_1(x,s) = P_{X_1}(x)P_{S_1}(s)$. 
Then, we show (see Theorems~\ref{thm:DP-fin} and~\ref{thm:DP-inf} below) that there is no loss of
optimality in restricting attention to charging policies of the form
\begin{equation} \label{eq:pi-policy}
  q_t(y_t | x_t, s_t, \pi_t).
\end{equation}
Such a charging policy is Markovian in the belief state~$\pi_t$ and the
optimal policies of such form can be searched using a dynamic program. 

To describe such a dynamic program, we assume that there is a decision maker
that observes $y^{t-1}$ (or equivalently $\pi_t$) and chooses ``actions''
$a_t = q_t(\cdot | \cdot, \cdot, \pi_t)$ using some decision rule $a_t =
f_t(\pi_t)$. We then identify a dynamic program to choose the optimal
decision rules.  

Note that the actions $a_t$ take vales in a subset $\ALPHABET A$ of
$\ALPHABET P_{Y|X,S}$ given by
\begin{multline} \label{eq:A}
  \ALPHABET A = \big\{ a \in \ALPHABET P_{Y|X,S}: a(\ALPHABET Y_\circ(s -
  x)\mid x,s) = 1, \\
  \forall (x,s) \in \ALPHABET X \times \ALPHABET S\big\}.
\end{multline}
\edit{The set $\ALPHABET A$ is convex and compact.\footnote{\edit{If
    $a_1, a_2 \in \ALPHABET A$, then any linear combination $a'$ of them,
    where $a' = \lambda a_1 + (1 -\lambda) a_2$ with $\lambda \in (0, 1)$,
    also satisfies $a'(\ALPHABET Y_\circ(s-x \mid x,s) = 1$ for all $(x,s) \in
  \ALPHABET X \times \ALPHABET S$ and, therefore, belongs to $\ALPHABET A$.
  Hence $\ALPHABET A$ is convex. Moreover, it is easy to see that $\ALPHABET
  A$ is closed and $\ALPHABET P_{Y|X,s}$ is compact. Therefore, $\ALPHABET
  A$ is also compact.}}}

To \edit{succinctly} write the dynamic program, 
for any $a \in \ALPHABET A$, we define the Bellman operator
$\mathscr B_a \colon [ \ALPHABET P_{X,S} \to \reals] \to [ \ALPHABET P_{X,S}
\to \reals]$ as follows: \edit{for any $V \colon \ALPHABET P_{X,S} \to
  \reals$  and
any}  $\pi \in \ALPHABET P_{X,S}$,
\begin{multline} \label{eq:bellman}
    [ \mathscr B_a V ](\pi) = I(a ; \pi)  \\ 
     + \sum_{\substack{ x \in \ALPHABET X, s \in \ALPHABET S, \\ y \in \ALPHABET Y }}
    \pi(x,s) a(y \mid x,s) V(\varphi(\pi, y, a))
\end{multline}
where the function $\varphi$ is a non-linear filtering equation defined in
Sec.~\ref{sec:DPdetails}.

Our main results are the following:
\edit{\begin{theorem}\label{thm:DP-fin}
  We have the following for Problem~\ref{prob:original} with a finite
  planning horizon~$T$:
  \begin{enumerate}
    \item \emph{Value functions}: Iteratively define \emph{value functions}
      $V_t \colon \ALPHABET P_{X,S} \to \reals$ as follows. For any $\pi \in
      \ALPHABET P_{X,S}$, $V_{T+1}(\pi) = 0$, and for $t \in \{T, \dots,
      1\}$,
      \begin{equation} \label{eq:fin-DP}
        V_t(\pi) = \min_{a \in \ALPHABET A} 
        [ \mathscr B_a V_{t+1}](\pi),
        \quad \forall \pi \in \ALPHABET P_{X,S}.
      \end{equation}
      Then, $V_t(\pi)$ are is continuous and concave in $\pi$. 
    \item \emph{Optimal policy}:
      Let $f^*_t(\pi)$ denote the arg min of the right hand side
      of~\eqref{eq:fin-DP}. Then, optimal policy $\mathbf q^* = (q^*_1, \dots,
      q^*_T)$ is given by
      \begin{equation*}
        q^*_t(y_t \mid x_t, s_t, \pi_t) = a_t(y_t \mid x_t, s_t),
        \text{  where } a_t = f^*_t(\pi_t).
      \end{equation*}
      Thus, there is no loss of optimality in restricting attention to
      charging policies of the form~\eqref{eq:pi-policy}.
    \item \emph{Optimal leakage rate}:
      The optimal (finite horizon) leakage rate is given by
      $V_1(\pi_1)/T$, where $\pi_1(x,s) = P_{X_1}(x) P_{S_1}(s)$. 
      \qed
  \end{enumerate}
\end{theorem}}
See Section~\ref{sec:DP} for proof.

\edit{
\begin{theorem} \label{thm:DP-inf}
  We have the following for Problem~\ref{prob:original} with an infinite
  planning horizon:
  \begin{enumerate}
    \item \emph{Value function}:
      Consider the following fixed point equation
      \begin{equation} \label{eq:inf-DP}
        J + v(\pi) = \min_{a \in \ALPHABET A }
        [\mathscr B_a v](\pi),
        \quad \forall \pi \in \ALPHABET P_{X,S},
      \end{equation}
      where $J \in \reals$ is a constant and $v \colon \ALPHABET P_{X,S} \to
      \reals$.
    \item \emph{Optimal policy}:
      Suppose there exists $(J,v)$ that satisfy~\eqref{eq:inf-DP}. 
           Let $f^*(\pi)$ denote the arg min of the right hand side
      of~\eqref{eq:inf-DP}. Then, the time-homogeneous optimal policy $\mathbf
      q^* = (q^*, q^*, \dots)$ given by 
      \[
        q^*(y_t \mid x_t, s_t, \pi_t) = a(y_t \mid x_t, s_t),
        \text{  where } a = f^*(\pi)
      \]
      is optimal. Thus, there is no loss of optimality in restricting
      attention to charging time-homogeneous policies of the
      form~\eqref{eq:pi-policy}.
    \item \emph{Optimal leakage rate}: The optimal (infinite horizon) leakage
      rate is given by $J$. 
      \qed
  \end{enumerate}
\end{theorem}}

\begin{proof}
\edit{Given the result of Theorem~\ref{thm:DP-fin}, the result of
Theorem~\ref{thm:DP-inf} follows from standard dynamic programming
arguments. See, for example, \cite[Theorem~5.2.4 and Eq.~(5.2.18)]{hl:MDP}.}
\end{proof}

\edit{There are various conditions that guarantee the existence of a $(J,v)$
  that satisfies~\eqref{eq:inf-DP}. Most of these conditions require the
  ergodicity of the process $\{\Pi_t\}_{t \ge 1}$ for every standary Markov
  procily $f \colon \ALPHABET P_{X,S} \to \ALPHABET A$. We refer the reader
  to~\cite[Chapter~3]{hl:AdaptiveMDP} and~\cite[Chapter 10]{hl:Topics} for
  more details.}

The dynamic program above resembles the dynamic program for partially
observable Markov decision processes (POMDP)
\edit{(see~\cite[Chapter~5]{bertsekas})} with hidden state~$(X_t, S_t)$,
observation~$Y_t$, and action~$A_t$. However, in contrast to POMDPs, the
expected per-step cost $I(a;\pi)$ is not linear in~$\pi$. Nonetheless, one
could use computational techniques from POMDPs to approximately solve the
dynamic programs of Theorems~\ref{thm:DP-fin} and~\ref{thm:DP-inf}. See
Section~\ref{sec:numerical} for a
brief discussion.

\begin{remark}
  \edit{Although the above results assume that the 
    demand is a first-order Markov chain, they extend
    naturally to the case when the demand is higher-order Markov. In
    particular, suppose the demand $\{X_t\}_{t \ge 1}$ is a $k$-th order
    Markov chain. Then, we can define another process $\{\tilde X_t\}_{t \ge
    1}$ where $\tilde X_t = (X_{t-k+1}, \dots, X_t)$, and use $\tilde X_t$ in
    Theorems~\ref{thm:DP-fin} and~\ref{thm:DP-inf}.}
\end{remark}

\subsection{Main result for i.i.d.\ demand}

Assume the following:
\begin{assumption}
  The demand $\{X_t\}_{t \ge 1}$ is i.i.d.\ with probability
    distribution~$P_X$. 
\end{assumption}
We provide an explicit characterization of optimal policy and optimal leakage
rate for this case. 

Define an auxiliary state variable $W_t = S_t - X_t$ that takes values in
$\ALPHABET W = \{ s - x : s \in \ALPHABET S, x \in \ALPHABET X \}$. 
For any $w \in \mathcal W$, define:
\begin{equation}\label{eq:D}
    \ALPHABET D(w) = \{ (x, s) \in \ALPHABET X \times \ALPHABET S : s - x = w \}.
\end{equation}
Then, we have the following.
\begin{theorem} \label{thm:iid}
  Define
  \begin{equation} \label{eq:opt}
    J^* = \min_{ \theta \in \ALPHABET P_S} I(S - X; X)
    = \edit{\min_{\theta \in \ALPHABET P_S} \big\{H(S-X) - H(S) \big\}}
  \end{equation}
  where $X$ and $S$ are independent with $X \sim P_X$ and $S \sim \theta$. Let
  $\theta^*$ denote the arg min in~\eqref{eq:opt}. Define $\xi^*(w) =
  \sum_{(x,s) \in \ALPHABET D(w)} P_X(x) \theta^*(s)$. Then, under (A)
  \begin{enumerate}
    \item $J^*$ is the optimal (infinite horizon) leakage rate.
    \item Define $b^* \in \ALPHABET P_{Y|W}$ as follows:
      \begin{equation}  \label{eq:bstar}
        b^*(y|w) = 
        \begin{cases} 
           P_X(y) \frac{\theta^*(y+w) } { \xi^*(w)}	& \text{ if } y \in \mathcal X \cap \mathcal Y_\circ(w)\\
          0       	& \text{ otherwise }.
        \end{cases}	
      \end{equation} 
      \edit{We call $b^*$ as a structured policy with respect to $(\theta^*,
      \xi^*)$.} 
      Then,  the memoryless charging policy $\mathbf q^* = (q^*_1, q^*_2,
      \dots)$ given by
      \begin{equation}  \label{eq:qstar}
        q^*_t(y \mid x, s, \pi_t) = b^*(y \mid s- x)
      \end{equation} 
      is optimal and achieves the optimal (infinite horizon) leakage rate.
      \qed
  \end{enumerate}
\end{theorem}

Note that the optimal charging policy is \emph{time-invariant} and \emph{memoryless}, i.e., the
distribution on $Y_t$ does not depend on $\pi_t$ (and, therefore on
$y^{t-1}$).

The proof, which is presented in Section~\ref{sec:iid}, is based on the standard arguments of showing achievability and a
converse. On the achievability side we show that the policy in~\eqref{eq:bstar}
belongs to a class of policies that satisfies a certain invariance property. Using
this property the multi-letter mutual information expression can be reduced into
a single-letter expression. For the converse
 we provide two proofs. The first is based on a
simplification of the dynamic program of Theorem~\ref{thm:DP-inf}. The second is
based on purely probabilistic and information theoretic arguments.

\subsection{Binary Model (Revisited)} \label{sec:revisit}

We revisit the binary model in Section~\ref{subsec:bin-model}, \edit{where the
  demand has equiprobable distribution, i.e., 
  ${P_X(x) = \frac 12}$ for $x \in \{0, 1\}$. Consider $\theta(0) = p$, $\theta(1) =
  1-p$. Let $W = S - X$. Then, 
  \begin{equation} \label{eq:W-bin}
    \PR(W = w) = \begin{cases}
      \frac 12 p, & \hbox{if $w = -1$} \\
      \frac 12 , & \hbox{if $w = 0$} \\
      \frac 12 (1-p), & \hbox{if $w = 1$}.
    \end{cases}
  \end{equation}
  Thus, 
  \[
    I(W; X) = H(W) - H(S) = 1 - \tfrac 12 h(p)
  \]
  where $h(p)$ is the binary entropy function: 
  \[
    h(p) = - p \log p - (1-p) \log(1-p).
  \]
  Thus, the value of $p$ that minimizes $I(W;X)$ is $p^* = \frac 12$ and the
  optimal leakage rate is $I(W;X) = \frac 12$.}
  
  \edit{The corresponding $\xi^*$ is obtained by substituting ${p=\frac 12}$
    in~\eqref{eq:W-bin}. For ease of notation, we denote $b^*(\cdot | w)$ as 
    $[ b^*(0|w),  b^*(1|w) ]$. Then, 
    \[
      b^*(\cdot | -1) = [0, 1], \quad
      b^*(\cdot| 0) = [\tfrac 12, \tfrac 12], \quad
      b^*(\cdot | -1) = [1, 0].
    \]
    It can be shown that this strategy is the same as~\eqref{eq:varodayan},
  which was proposed in~\cite{varodayan}.}
This yields an analytical proof of the result in~\cite{varodayan}.

\subsection{Salient features of the result for i.i.d.\ demand}

Theorem~\ref{thm:iid} shows that even if consumption could take larger values
than the demand, i.e., $\ALPHABET Y \supset \ALPHABET X$, under \edit{the
optimal} policy, $Y_t$
takes values only in $\ALPHABET X$. This agrees with the intuition that a
consumption larger that $m_x$ reveals that the battery has low charge and
that the power demand is high. In extreme cases, a large consumption may
completely reveal the battery and power usage thereby increasing the
information leakage. 

We now show some other properties of the optimal policy.
% \begin{property} \label{prop:diff}
%   The mutual information $I(S - X; X)$ is equal to $H(S-X) - H(S)$.
% \end{property}
% \begin{proof}
%   This follows from the following simplifications:
%   \begin{align}
%     I(S - X; X) &= H(S - X) - H(S - X | X) \notag \\
%     &= H(S - X) - H(S | X) \notag \\
%     &= H(S - X) - H(S) \tag*{\qedhere}
%   \end{align}
% \end{proof}
\begin{property} \label{prop:convex}
  The mutual information $I(S - X; X)$ is strictly convex in the distribution
  $\theta$ and, therefore, $\theta^* \in \mathit{int}(\ALPHABET P_S)$.
\end{property}
See Appendix~\ref{app:convex} for proof.

As a consequence, the optimal $\theta^*$ in~\eqref{eq:opt} may be obtained
using the Blahut-Arimoto algorithm~\cite{Blahut:1972, Arimoto:1972}.

\begin{property}
  Under the battery charging policy specified in Theorem~\ref{thm:iid}, 
  the power consumption $\{Y_t\}_{t \ge 1}$ is i.i.d.\ with marginal
  distribution $P_X$. Thus, $\{Y_t\}_{t \ge 1}$ is statistically
  indistinguishable from $\{X_t\}_{t \ge 1}$.
\end{property}
See Remarks~\ref{rem:iid} and~\ref{rem:PY=PX} in
Section~\ref{sec:achievability} for proof.

\begin{property} \label{prop:sym}
  If the power demand has a symmetric PMF, i.e., for any $x \in \ALPHABET
  X$, $P_X(x) = P_X(m_x - x)$, then the optimal $\theta^*$ in
  Theorem~\ref{thm:iid} is also symmetric, i.e., for any $s \in \ALPHABET S$,
  $\theta^*(s) = \theta^*(m_s - s)$.
\end{property}

\begin{proof}
  For $\theta \in \mathcal P_S$, define ${\bar \theta(s) =
  \theta(m_s - s)}$. Let $X \sim P_X$, $S \sim \theta$ and $\bar S \sim \bar
  \theta$. Then, by symmetry
  \begin{equation} \label{eq:sym}
    I(S - X; X) = I(\bar S - X; X).
  \end{equation}
  For any $\lambda \in (0,1)$, let $\theta_\lambda(s) = \lambda \theta(s) +
  (1-\lambda) \bar \theta(s)$ denote the convex combination of $\theta$ and
  $\bar \theta$. Let $S_\lambda \sim \theta_\lambda$. By
  Property~\ref{prop:convex}, if $\theta \neq \bar \theta$, then
\begin{align*}
  I(S_\lambda - X; X) &< \lambda I(S - X; X) + (1-\lambda) I(S-X; X) \\
  &= I(S - X; X),
\end{align*}
where the last equation uses~\eqref{eq:sym}.

Thus, if $\theta \neq \bar \theta$, we can strictly decrease the mutual
information by using $\theta_\lambda$. Hence, the optimal distribution must
have the property that $\theta^*(s) =  \theta^*(m_s - s)$.
\end{proof}

Given a distribution $\mu$ on some alphabet $\ALPHABET M$, we say that the
distribution is \emph{\textbf{almost} symmetric and unimodal}
around $m^* \in \ALPHABET M$ if 
\[
  \mu_{m^*} \ge \mu_{m^* + 1} \ge \mu_{m^* - 1} \ge \mu_{m^* + 2} \ge
  \mu_{m^* -2 } \ge \dots
\]
where we use the interpretation that for $m \not\in \ALPHABET M$, $\mu_m = 0$.
Similarly, we say that the distribution is \emph{symmetric and
unimodal} around $m^* \in \ALPHABET M$ if
\[
  \mu_{m^*} \ge \mu_{m^* + 1} = \mu_{m^* - 1} \ge \mu_{m^* + 2} =
    \mu_{m^* - 2} \ge \dots
\]
Note that a distribution can be symmetric and unimodal only if its support is
odd.

\begin{property} \label{prop:unimodal}
  If the power demand is symmetric and unimodal around $\lfloor m_x/2
  \rfloor$, then the optimal $\theta^*$ in Theorem~\ref{thm:iid} is
  \emph{almost} symmetric \edit{and unimodal  around} $\lfloor m_s/2 \rfloor$.
  In particular, if $m_s$ is even, then 
  \[
    \theta^*_{m^*} \ge \theta^*_{m^* + 1} = \theta^*_{m^* - 1} 
    \ge \theta^*_{m^* + 2} = \theta^*_{m^* - 2} \ge \dots
  \]
  and if $m_s$ is odd, then
  \[
    \theta^*_{m^*} = \theta^*_{m^* + 1} \ge \theta^*_{m^* - 1} 
    = \theta^*_{m^* + 2} \ge \theta^*_{m^* - 2} = \dots
  \]
  where $m^* = \lfloor m_s/2 \rfloor$.
\end{property}
\begin{proof}
  Let $\bar X = - X$. Then, $I(S-X;S) = H(S-X) - H(S) = H(S+\bar X) -
  H(S)$. Note that $P_{\bar X}$ is also symmetric and
  unimodal around $\lfloor m_x/2 \rfloor$.

  Let $S^\circ$ and $\bar X^\circ$ denote the random variables $S - \lfloor
  m_s/2 \rfloor$ and $\bar X - \lfloor m_x/2 \rfloor$. Then $\bar X^\circ$
  is also symmetric and unimodal around \edit{the} origin and 
  \[
    I(S - X; X) = H(S^\circ + \bar X^\circ) - H(S^\circ).
  \]

  Now given any distribution $\theta^\circ$ of $S^\circ$, let $\theta^+$ be a
  permutation of $\theta^\circ$ that is almost symmetric and unimodal with a
  positive bias around origin. Then by~\cite[Corollary III.2]{madiman},
  $H(P_X * \theta^\circ) \ge H(P_X * \theta^+)$. Thus, the optimal
  distribution must have the property that $\theta^\circ = \theta^+$ or,
  equivalently, $\theta$ is almost unimodal and symmetric around $\lfloor
  m_s/2 \rfloor$.

  Combining this with the result of Property~\ref{prop:sym} gives the
  characterization of the distribution when $m_s$ is even or odd.
\end{proof}

\subsection{Numerical Example: i.i.d.\ demand}

Suppose there are $n$ identical devices in the house and each is on with
probability~$p$. Thus, $X \sim \text{Binomial}(n, p)$. We derive the optimal
policy and optimal leakage rate for this scenario under the assumption that
$\ALPHABET Y = \ALPHABET X$. We consider two specific examples, where we
numerically solve~\eqref{eq:opt}.

Suppose $n = 6$ and $p=0.5$.
\begin{enumerate}
  \item Consider $\ALPHABET S = [0{:}5]$. Then, by numerically
    solving~\eqref{eq:opt}, we get that the optimal leakage rate $J^*$ is is
    $0.4616$ and the optimal battery charge distribution $\theta^*$ is
  \[
    \{    
       0.1032,    0.1747,    0.2221,    0.2221,    0.1747,    0.1032
    \}.
  \]
  \item Consider $\ALPHABET S = [0{:}6]$. Then, by numerically
    solving~\eqref{eq:opt}, we get that the optimal leakage rate $J^*$ is is
    $0.3774$ and the optimal battery charge distribution $\theta^*$ is
  \[
    \{    
       0.0773,    0.1364,    0.1847,    0.2031,    0.1847,    0.1364,    0.0773
    \}.
  \]
\end{enumerate}
Note that both results are consistent with Properties~\ref{prop:sym}
and~\ref{prop:unimodal}.

\begin{figure}[t]
  \centering
  \includegraphics[scale=0.625, trim = {0.9cm 0.3cm 1.4cm 0.3cm},clip]{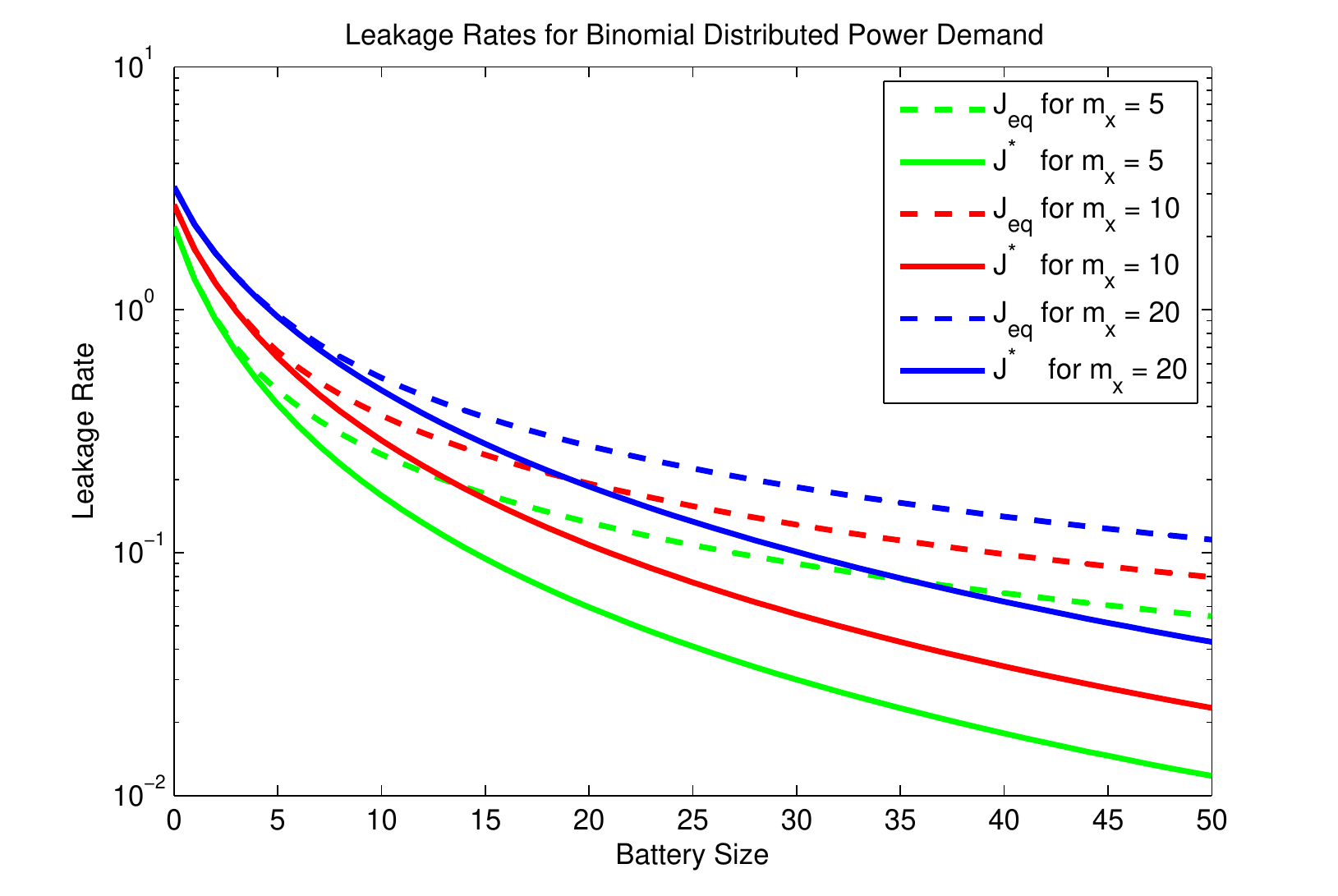}
  \caption{\small{A comparison of the performance of $q_{eq} \in Q_B$ as
  defined in~\eqref{eq:qeq} with the optimal leakage rate for i.i.d.\ Binomial distributed demand Binomial($m_x$, 0.5) for $m_x = \{5,10,20\}$.}}
  \label{largeBattPlot}
\end{figure}

We next compare the performance with the following time-homogeneous benchmark
policy $\mathbf q_{eq} \in \ALPHABET Q_B$: for all $y \in \mathcal Y, w \in
\mathcal W$, 
\begin{equation} \label{eq:qeq}
  q_t(y_t | w_t) = 
  \frac {\IND_{\mathcal Y_\circ(w_t)}\{y_t\}}{|\mathcal Y_\circ(w_t)|}.
\end{equation}
This benchmark policy chooses all feasible values of $Y_t$ with equal
probability. For that reason we call it \emph{equi-probably policy} and
denote its performance by $J_{eq}$.

In Fig. \ref{largeBattPlot}, we compare the performance of $\mathbf q^*$ and
$\mathbf q_{eq}$ as a function of battery sizes for different demand
alphabets.

Under $\mathbf q_{eq}$, the MDP converges to a belief state that is
approximately uniform. Hence, in the low battery size regime, $J_{eq}$ is
close to optimal but its performance gradually worsens with increasing
battery size.

\edit{\subsection{Numerical Example: Continuous Valued Alphabets}}

\edit{Although our setup assumes discrete alphabets, it can shown that Theorem~\ref{thm:iid} can be
extended to continuous alphabets under mild
technical conditions assuming that the density function for $X$, say $f_X(x)$ exists;
see~\cite{simonThesis16} for details.  As such we provide two achievability proofs
for Theorem~\ref{thm:iid}. The {\em weak achievability}  in Section~\ref{sec:achievability} assumes that the distribution of the initial state $S_1$
can be selected by the user. The {\em strong achievability} in Section~\ref{sec:strong_achievability} assumes that the initial state $S_1$ can be arbitrary.
The proof of the weak achievability extends immediately to continuous valued alphabets. Furthermore the proof of the converse based on information theoretic
arguments in Section~\ref{subsec:IT_Conv} also extends to continuous valued alphabets. We provide a numerical example involving a continuous valued input.}

\label{subsec:cont}

\edit{Let ${\mathcal X} = [0,1]$ be continuous valued, and let $f_X(x) = 1$
  for $x \in [0,1]$. We assume that ${\mathcal S} = [0,B]$ where $B$ denotes
  the storage capacity. We will assume that $B \ge 2$ for convenience.
  Following Theorem~\ref{thm:iid}, it suffices to take the output alphabet to
  be ${\mathcal Y} = [0,1]$. Note that for $W = S-X$, we have that
  ${\mathcal W} = [-1, B]$ and that the support of the output for any given $w$ is given by:
\begin{align}
{\mathcal Y_\circ(w)} = \begin{cases}
[-w, 1], & -1 \le w \le 0, \\
[0,1], & 0 \le w \le B-1, \\
[0, B-w], & B-1\le w \le B.
\end{cases}
\label{eq:y0_df}
\end{align}
Let $\theta^\star(\cdot)$ be the  density function that minimizes~\eqref{eq:opt} and $S^\star$ denote the random variable with this density, independent of $X$. Let $W = S^\star - X$ and  $\xi^\star(w)$ be the associated density. Then it follows from~\eqref{eq:bstar} that the optimal policy is given by
\begin{align}
b^\star(y|w) = \frac{\theta^\star(y+w)}{\xi^\star(w)}, \quad y \in {\mathcal Y_\circ(w)}. \label{eq:bs2}
\end{align}
Instead of  computing $\theta^\star(\cdot)$ numerically, which is cumbersome due to the density functions, we provide an analytical lower bound on the leakage rate. Note that the objective in~\eqref{eq:opt} can be expressed as:
\begin{align}
I(S-X; X) = h(S-X) - h(S) \label{eq:mut_inf_Exp}
\end{align}
where $h(\cdot)$ is the differential entropy. Using the entropy power inequality~\cite{cover}, since $S$ and $X$ are independent, we have
for any density $f_S(\cdot)$ that: 
\begin{align}
2^{2 h(S-X)} \ge 2^{2h(S)} + 2^{2h(X)}, \label{eq:epi}
\end{align}
where we use the fact that $h(-X) = h(X)$. Substituting in~\eqref{eq:mut_inf_Exp} we have
\begin{align}
I(S-X; X) &\ge \frac{1}{2}\log_2\left(2^{2h(S)} + 2^{2h(X)} \right) - h(S) \notag \\
&= \frac{1}{2}\log_2\left(1+  2^{2h(X) - 2 h(S)} \right) \notag\\
&=\frac{1}{2} \log_2\left(1+  2^{- 2 h(S)} \right)  \label{eq:hx}\\
&\ge \min_{S} \frac{1}{2} \log_2\left(1+  2^{- 2 h(S)} \right)\notag\\
&=  \frac{1}{2} \log_2\left(1+  2^{- 2 \{ {\max_{S} h(S)\}}} \right)\label{eq:hx2}\\
&= \frac{1}{2}\log_2\left(1+ \frac{1}{B^2}\right) \label{eq:hx3}
\end{align}
where we use the fact that when $X \sim \mathrm{Unif}[0,1]$, we have that $h(X) = 0$, see e.g.,~\cite{cover2} in~\eqref{eq:hx}
and the fact that the expression in~\eqref{eq:hx2} is decreasing in $h(S)$.  Finally in~\eqref{eq:hx3} we use the fact that a uniform
distribution maximizes the differential entropy when ${\mathcal S} = [0,B]$ is fixed and the maximum value is $h(S) = \log_2 B$.
\endgraf
For the achievability, we evaluate the leakage rate by selecting $\theta^\star(s)$ to be a uniform distribution over $[0,B]$
and using~\eqref{eq:bs2}. The resulting leakage rate is given by
\begin{align}
L^+ = h(W) - h(S) = h(\xi) - \log_2(B)
\end{align}
where the density $\xi(w)$ for $W = S-X$ is as follows:
\begin{align*}
\xi(w) = \begin{cases}
\frac{1+w}{B}, & -1 \le w \le 0 \\
\frac{1}{B}, & 0 \le w \le B-1 \\
\frac{B-w}{B}, & B-1 \le w \le B
\end{cases}
\end{align*}
and we use that when $\theta(s)$ is a uniform density over $[0,B]$, it follows that $h(S)=\log_2(B)$. 
Through straightforward computations it can be shown that 
\begin{align}h(\xi) = \frac{1}{2B \ln 2}  + \log_2(B)\end{align} and thus it follows that
$$L^+ = \frac{1}{2B\ln 2}$$
is achievable with a uniform distribution on the state. 
We note that the analytical lower bound on the leakage rate in~\eqref{eq:hx3} decays as $1/B^2$ for large $B$ while the achievable
rate decays as $1/B$. It will be interesting in future to determine the structure of the optimal input distribution
and study the associated leakage rate.}

\section {Proof of Theorem~\ref{thm:DP-fin}} \label{sec:DP}

One of the difficulties in obtaining a dynamic programming decomposition for
Problem~\ref{prob:original} is that the objective function is not of the form
$\sum_{t=1}^T \text{cost}(\text{state}_t, \text{action}_t)$. We show that
there is no loss of optimality to restrict attention to a class of policies
$\ALPHABET Q_B$ and for any policy in $\ALPHABET Q_B$, the mutual information
may be written in an additive form. 

\subsection {Simplification of optimal charging policies} \label{sec:A-B}

Let $\ALPHABET Q_B \subset \ALPHABET Q_A$ denote the set of charging
policies that choose consumption xd only on the consumption history,
current demand, and battery state. Thus, for $\mathbf{q} \in \ALPHABET Q_B$,
at any time t, given history $(x^t$, $s^t$, $y^{t-1})$, the consumption $Y_t$
is $y$ with probability $q_t(y \mid x_t, s_t, y^{t-1})$. Then the joint
distribution on $(X^T, S^T, Y^T)$ induced by $\mathbf{q} \in \mathcal Q_B$ is
given by  
\begin{multline*} 
  \PR^{\mathbf{q}}( S^{T} = s^{T}, X^{T} = x^{T}, Y^{T} = y^{T}) \\
  = P_{S_1}(s_1) P_{X_1}(x_1) q_1(y_1 \mid x_1, s_1)  \prod_{t=2}^{T} \bigg[ \IND_{s_{t}}\{s_{t-1} - x_{t-1} + y_{t-1} \} \\
    \times Q(x_{t}| x_{t-1}) q_t(y_t \mid x_{t}, s_{t}, y^{t-1}) 
  \bigg].
\end{multline*}

\begin{proposition} \label{prop:A-B}
  In Problem~\ref{prob:original}, there is no loss of optimality in
  restricting attention to charging policies in $\ALPHABET Q_B$.  Moreover,
  for any $\mathbf{q} \in \ALPHABET Q_B$, the objective function takes an
  additive form:
  \begin{equation*} %\label{eq:additive}
    L_T(\mathbf q) = \frac 1T \sum_{t=1}^T I^{\mathbf q}(X_t, S_t; Y_t \mid Y^{t-1})   
  \end{equation*}
  where
  \begin{align*}
    &I^\mathbf{q}(X_t, S_t; Y_t \mid Y^{t-1}) \\
    &\qquad = \sum_{ \substack{
      x_t \in \ALPHABET X,  s_t \in \ALPHABET S \\
      y^t \in \ALPHABET Y^t 
    }}
    \begin{lgathered}[t]
      \PR^\mathbf{q}(X_t = x_t, S_t = s_t, Y^t = y^t)
      \\
      \times
      \log \frac {q_t(y_t \mid x_t, s_t, y^{t-1})}
      {\PR^\mathbf{q}(Y_t = y_t \mid Y^{t-1} = y^{t-1})}.
    \end{lgathered}
  \end{align*} 
\end{proposition}

See Appendix~\ref{app:A-B} for proof. The intuition behind why policies in
$\ALPHABET Q_B$ are \edit{no worse that others} in $\ALPHABET Q_A$ is as follows. For a
policy \edit{in} $\mathcal Q_A$, observing the realization $y^{t}$
of $Y^{t}$ gives partial information about the history $(x^{t}, s^{t})$
while for a policy \edit{in} $\mathcal Q_B$, $y_t$ gives partial
information only about the current state $(x_t, s_t)$. The dependence on
$(x_t, s_t)$ cannot be removed because of the conservation
constraint~\eqref{eq:conservation}. 

Proposition~\ref{prop:A-B} shows that the total cost may be written in an
additive form. Next we use an approach inspired
by~\cite{Tatikonda09,goldsmith,permuter} and formulate an equivalent
sequential optimization problem.

\subsection {An equivalent sequential optimization problem} \label{sec:equiv}

Consider a system with state process $\{ X_t, S_t \}_{t\geq 1}$ where
$\{X_t\}_{t\geq 1}$ is an exogenous Markov process as before and
$\{S_t\}_{t\geq 1}$ is a controlled Markov process as specified below. At
time~$t$, a decision maker observes $Y^{t-1}$ and chooses a distribution valued
action $A_t \in \ALPHABET A$, where $\ALPHABET A$ is given by~\eqref{eq:A},
as follows:
\begin{equation} \label{eq:f-policy}
  A_t = f_t(Y^{t-1})
\end{equation}
where $\mathbf{f} = (f_1, f_2, \dots)$ is called the decision policy. 

Based on this action, an auxiliary variable $Y_t \in \ALPHABET Y$ is chosen
according to the conditional probability $a_t(\cdot \mid x_t, s_t)$ and the
state $S_{t+1}$ evolves according to~\eqref{eq:conservation}. 

At each stage, the system incurs a per-step cost given by
\begin{equation}
  c_t(x_t, s_t, a_t, y^{t}; \mathbf f) := 
  \log \frac {a_t(y_t \mid x_t, s_t)} 
  {\PR^{\mathbf f}(Y_t = y_t \mid Y^{t-1} = y^{t-1})}.
  \label{eq:cost}
\end{equation}

The objective is to choose a policy $\mathbf f = (f_1, \dots, f_T)$ to
minimize the total finite horizon cost given by
\begin{equation} \label{eq:alt-cost}
  \tilde L_T(\mathbf f) := \frac 1T \EXP^{\mathbf f}\left[ \sum_{t=1}^T 
  c_t(X_t, S_t, A_t, Y^{t}; \mathbf f) \right]
\end{equation}
where the expectation is evaluated with respect to the
probability distribution $\PR^{\mathbf f}$ induced by the decision policy
$\mathbf f$.

\begin{proposition} \label{prop:equiv}
  The sequential decision problem described above is equivalent to
  Problem~\ref{prob:original}. In particular,
  \begin{enumerate}
    \item Given $\mathbf{q} = (q_1, \dots, q_T) \in \ALPHABET Q_B$, let
      $\mathbf f = (f_1, \dots, f_T)$ be  
      \[
        f_t(y^{t-1}) = q_t( \cdot \mid \cdot, \cdot, y^{t-1}).
      \]
      Then $\tilde L_T(\mathbf f) = L_T(\mathbf{q})$.

    \item Given $\mathbf f = (f_1, \dots, f_T)$, let $\mathbf{q} = (q_1,
      \dots, q_T) \in \ALPHABET Q_B$ be 
      \[
        q_t(y_t \mid x_t, s_t, y^{t-1}) = a_t(y_t \mid x_t, s_t),
        \text{ where } a_t = f_t(y^{t-1}).
      \]
      Then $L_T(\mathbf{q}) = \tilde L_T(\mathbf f)$.
  \end{enumerate}
\end{proposition}
\begin{proof}
For any
history $(x^t$, $s^t$, $y^{t-1})$, $a^t \in \mathcal A$, and $s_{t+1} \in
\ALPHABET S$,
\begin{align}
  \hskip 2em & \hskip -2em
  \PR(S_{t+1} = s_{t+1} \mid X^t = x^t, S^t = s^t, Y^{t} =
  y^{t}, A^t = a^t) \notag \\
  &= \sum_{y_t \in \ALPHABET Y} 
  \IND_{s_{t+1}}\left\{ s_t + y_t - x_t \right\}
  a_t(y_t \mid x_t, s_t) 
  \notag \\
  &= \PR(S_{t+1} = s_{t+1} \mid X_t = x_t, S_t = s_t, A_t = a_t).
  \label{eq:CMDP}
\end{align}
Thus, the probability distribution on $(X^T, S^T, Y^T)$ induced by a decision
policy $\mathbf f = (f_1, \dots, f_T)$ is given by
\begin{multline*} 
  \PR^{\mathbf{f}}( S^{T} = s^{T}, X^{T} = x^{T}, Y^{T} = y^{T}) \\
  = P_{S_1}(s_1) P_{X_1}(x_1) q_1(y_1 \mid x_1, s_1)  \prod_{t=2}^{T} \bigg[ \IND_{s_{t}}\{s_{t-1} - x_{t-1} + y_{t-1} \} \\
    \times Q(x_{t}| x_{t-1}) a_t(y_t | x_{t}, s_{t}) 
  \bigg].
\end{multline*}
where $a_t = f_t(y^{t-1})$. Under the transformations described in
the Proposition, $\PR^{\mathbf f}$ and $\PR^\mathbf{q}$ are identical
probability distributions. Consequently, $\EXP^{\mathbf f}[ c_t(X_t, S_t,
A_t, Y^t; \mathbf f)] = I^\mathbf{q}(S_t, X_t; Y_t \mid Y^{t-1})$. Hence,
$L_T(\mathbf q)$ and $\tilde L_T(\mathbf f)$ are equivalent.
\end{proof}

Eq.~\eqref{eq:CMDP} implies that $\{X_t, S_t\}_{t \ge 1}$ is a controlled
Markov process with control action $\{A_t\}_{t \ge 1}$. In the next section,
we obtain a dynamic programming decomposition for this problem. For the
purpose of writing the dynamic program, it is more convenient to write the
policy~\eqref{eq:f-policy} as 
\begin{equation} \label{eq:f-policy2}
  A_t = f_t(Y^{t-1}, A^{t-1}).
\end{equation}
Note that these two representations are equivalent. Any policy of the
form~\eqref{eq:f-policy} is also a policy of the form~\eqref{eq:f-policy2}
(that simply ignores $A^{t-1}$); any policy of the form~\eqref{eq:f-policy2}
can be written as a policy of the form~\eqref{eq:f-policy} by recursively
substituting $A_t$ in terms of $Y^{t-1}$. Since the two forms are equivalent,
in the next section we assume that the policy is of the
form~\eqref{eq:f-policy2}.

\subsection{A dynamic programming decomposition} \label{sec:DPdetails}

The model described in Section~\ref{sec:equiv} above is similar to a POMDP
(partially observable Markov decision process): the system state $(X_t, S_t)$
is partially observed by a decision maker who chooses action $A_t$. However,
in contrast to the standard cost model used in POMDPs, the per-step cost
depends on the observation history and \emph{past policy}. Nonetheless, if we
consider the belief state as the information state, the problem can be
formulated as a standard MDP.

For that matter, for any realization $y^{t-1}$ of past observations and any
choice $a^{t-1}$ of past actions, define the belief state $\pi_t \in \ALPHABET
P_{X,S}$ as follows: For $s \in \ALPHABET S$ and  $x \in \ALPHABET X$, 
\[
  \pi_{t}(x,s) = 
  \PR^{\mathbf f}(X_{t} = x, S_{t} = s | Y^{t-1} = y^{t-1}, A^{t-1} = a^{t-1}).
\]
If $Y^{t-1}$ and $A^{t-1}$ are random variables, then the belief state is a
$\ALPHABET P_{X,S}$-valued random variable. 

The belief state evolves in a state-like manner as follows.
\begin{lemma}\label{lem:pi}
  For any realization $y_t$ of $Y_t$ and $a_t$ of $A_t$, $\pi_{t+1}$ is given
  by
  \begin{equation} \label{eq:pi}
    \pi_{t+1} = \varphi(\pi_t, y_t, a_t)
  \end{equation}
  where $\varphi$ is given by
  \begin{multline*} %\label{eq:update}
      \varphi(\pi, y, a)(x',s') \\
       =\frac
       {\sum_{x \in \ALPHABET X} Q(x'|x) 
         a(y|x, s' - x + y) \pi(x, s' - x + y)}
      {\sum_{(x,s) \in \ALPHABET X \times \ALPHABET S}
        a(y|x,s)\pi(x,s)}.
  \end{multline*}
\end{lemma}
\begin{proof}
  For ease of notation, we use $\PR(x_t, s_t | y^{t-1}, a^{t-1})$ to denote
  $\PR(X_t = x_t, S_t = s_t | Y^{t-1} = y^{t-1}, A^{t-1} = a^{t-1})$. Similar
  interpretations hold for other expressions as well. Consider
  \begin{align}
    \pi_{t+1}(x_{t+1}, s_{t+1}) &= \PR(x_{t+1}, s_{t+1} | y^t, a^t) \notag \\
    &= \frac {\PR(x_{t+1}, s_{t+1}, y_t, a_t | y^{t-1}, a^{t-1})}
             {\PR(                  y_t, a_t | y^{t-1}, a^{t-1})}
    \label{eq:pi-1}
  \end{align}
  Now, consider the numerator of the right hand side. 
  \begin{align}
    \PR(&x_{t+1}, s_{t+1}, y_t, a_t | y^{t-1}, a^{t-1}) \notag \\
    &= \PR(x_{t+1}, s_{t+1}, y_t, a_t | y^{t-1}, a^{t-1}, \pi_t) \notag \\
    &= \smashoperator[l]{\sum_{ (x_t, s_t) \in \ALPHABET X \times \ALPHABET S }}
    \PR(x_{t+1} | x_t) \IND_{ s_{t+1} }(s_t + x_t - y_t) \notag\\
    &\qquad \times 
    a_t(y_t | x_t, s_t) 
    \IND_{ a_t }(f_t(\edit{y^{t-1}, a^{t-1}})) \pi_t(x_t, s_t)
    \label{eq:pi-2}
  \end{align}

  Substituting~\eqref{eq:pi-2} in~\eqref{eq:pi-1} (and observing that
  the denominator of the right hand side of~\eqref{eq:pi-1} is the marginal
  of the numerator over $(x_{t+1}, s_{t+1})$), we get that $\pi_{t+1}$ can be
  written in terms of $\pi_t$, $y_t$ and $a_t$. Note that if the term
  $\IND_{a_t}(f_t(\edit{y^{t-1}, a^{t-1}}))$ is~$1$, it cancels from both the numerator and the
  denominator; if it is~$0$, we are conditioning on a null event
  in~\eqref{eq:pi-1}, so we can assign any valid distribution to the
  conditional probability. 
\end{proof}

Note that an immediate implication of the above result is that $\pi_t$
depends only on $(y^{t-1}, a^{t-1})$ and not on the policy $\mathbf f$.
This is the main reason that we are working with a policy of the
form~\eqref{eq:f-policy2} rather than~\eqref{eq:f-policy}.

\begin{lemma}\label{lem:cost}
  The cost $\tilde L_T(\mathbf f)$ in~\eqref{eq:alt-cost} can be written as
  \[
    \tilde L_T(\mathbf f) = \frac 1T \EXP \bigg[
      \sum_{t=1}^T I(\edit{A_t}; \edit{\Pi_t})
    \bigg]
  \]
  where $I(a_t; \pi_t)$ does not depend on the policy $\mathbf f$ and is
  computed according to the standard formula
  \[
    I(a_t;\pi_t) = \sum_{\substack{ x \in \ALPHABET X, s \in \ALPHABET S, \\ y \in \ALPHABET Y }}
    \begin{lgathered}[t]
      \pi_t(x,s) a_t(y \mid x,s) 
      \\ \times \log \frac{ a_t(y|x,s) }
      {\strut
	\smashoperator
	{\sum\limits_{(\tilde x,\tilde s) \in \ALPHABET X \times \ALPHABET S }}
      \pi_t(\tilde x,\tilde s) a_t(y \mid \tilde x,\tilde s) }.
    \end{lgathered}
  \]
\end{lemma}
\begin{proof}
  By the law of iterated expectations, we have
  \begin{equation} \label{eq:tilde-L}
    \tilde L_T(\mathbf f) = \frac 1T \EXP
    \bigg[ \sum_{t=1}^T \EXP[ c_t(X_t, S_t, A_t, Y^t; \mathbf f) | Y^{t-1},
    A^{t-1} ] \bigg]
  \end{equation}
  Now, from~\eqref{eq:cost}, each summand may be written as
  \begin{align*}
    &\EXP^{\mathbf f}[ c_t(X_t, S_t, A_t, Y^t; \mathbf f) \mid Y^{t-1} = y^{t-1}, A^{t} = a^{t} ]  \\
    &\quad = \sum_{\substack{ x \in \ALPHABET X, s \in \ALPHABET S, \\ y \in \ALPHABET Y }}
    \begin{lgathered}[t]
      \pi_t(x,s) a_t(y \mid x,s) 
      \\ \times \log \frac{ a_t(y|x,s) }
      {\strut
	\smashoperator
	{\sum\limits_{(\tilde x,\tilde s) \in \ALPHABET X \times \ALPHABET S }}
      \pi_t(\tilde x,\tilde s) a_t(y \mid \tilde x,\tilde s) }
    \end{lgathered}\\
    &\qquad = I(a_t; \pi_t).
  \end{align*}
  \edit{Thus, $\EXP^{\mathbf f}[ c_t(X_t, S_t, Y^t; \mathbf f \mid Y^{t-1}, A^t ] =
      I(A_t, \Pi_t)$. Substituting this back in~\eqref{eq:tilde-L}, we get the
    result of the Lemma.}
\end{proof}

\begin{proof}[Proof of Theorem~\ref{thm:DP-fin}]
  Lemma~\ref{lem:pi}  implies that $\{\Pi_t\}_{t \geq 1}$ is a controlled
  Markov process with control action $A_t$. In addition, Lemma~\ref{lem:cost}
  implies that the objective function can be expressed in terms of the
  \emph{state} $\Pi_t$ and the action $A_t$. \edit{Thus, one can use Markov decision
  theory~\cite{hl:MDP} to identify the optimal policy. Since
both the state space and the action space are continuous valued, we need to
verify the standard technical conditions.}

\edit{Define the stochastic kernel $K \colon \ALPHABET P_{X,S} \times
  \ALPHABET A \to \ALPHABET P_{X,S}$ as follows. For any Borel subset $B$ of
  $\ALPHABET P_{X,S}$ and any $\pi \in \ALPHABET P_{X,S}$ and $a \in \ALPHABET
  A$, 
  \[
    K(B\mid \pi, a) = 
    \sum_{\substack{ x \in \ALPHABET X, s \in \ALPHABET S, \\ y \in \ALPHABET Y }}
    \pi(x,s) a(y \mid x,s) \IND_{B}\{\varphi(\pi, y, a))\}.
  \]
  The sequential model of Sec.~\ref{sec:equiv} has the following properties.
  \begin{enumerate}
    \item Based on~\eqref{eq:CMDP}, we can write the controlled dynamics of
      the state $(X_t, S_t)$ as follows:
      \begin{multline*}
        \PR(X_{t+1} = x_+, S_{t+1} = s_+ \mid X_t = x, S_t = s, A_t = a) 
        \\
        = \sum_{y \in \ALPHABET Y} Q(x_+|x)a(y)\IND_{s_+}\{s + y - x \},
      \end{multline*}
      which is continuous in~$a$. 
    \item The observations are given by
      \[
        \PR(Y_t = y | X_t = x, S_t = s, A_t = a) = a(y)
      \]
      which is continuous in~$a$.
    \item The observations $Y_t$ are discrete.
  \end{enumerate}
  Therefore, from~\cite[Sec~4.4 and Lemma~4.1]{hl:AdaptiveMDP}, we get  the
  following:
  \begin{enumerate}
    \item[4)] The stochastic kernel $K(d\pi_+ \mid \pi, a)$ is weakly
    continuous. 
  \end{enumerate}
  In addition, the model has the following properties:
  \begin{enumerate}
    \item[5)] The action set $\ALPHABET A$ is compact.
    \item[6)] The per-step cost $I(a,\pi)$ is continuous and bounded below.
    (In fact, the per-step cost is also bounded above). 
  \end{enumerate}
  Properties 4)--6) imply~\cite[Condition 3.3.3]{hl:MDP}, which
  by~\cite[Theorem~3.3.5]{hl:MDP}, implies the \emph{measurable selection
  condition}~\cite[Assumption 3.3.1]{hl:MDP}. Under the measurable selection
  condtion, the ``inf'' in the dynamic program can be replaced by a ``min''.
  \endgraf
  The continuity of the value function in $\pi$ follows from the continuity of the
  per-step cost $I(a, \pi)$and the controlled stochastic kernel $K(d\pi_+ \mid
  \pi, a)$.  The concavity of the value functions is proved in
  Appendix~\ref{app:concave}.
  \endgraf
  From standard results in Markov decision theorem (e.g.,
  see~\cite[Theorem~3.2.1]{hl:MDP}), it follows that the policy given in
  part~2) of the theorem is optimal for the sequential model of
  Sec~\ref{sec:equiv}. 
  Proposition~\ref{prop:equiv} implies that this policy is also optimal for
  Problem~\ref{prob:original}. 
}
\end{proof}

\subsection{Remarks about numerical solution} \label{sec:numerical}

The dynamic programs of Theorems~\ref{thm:DP-fin} and~\ref{thm:DP-inf}, both state and action spaces
are distribution valued (and, therefore, subsets of Euclidean space).
Although, an exact solution of the dynamic program is not possible, there are
two approaches to obtain an approximate solution. The first is to treat it as
a dynamic program of an MDP with continuous state and action spaces and use
approximate dynamic programming~\cite{powell,bertsekas}. The second is to
treat it as a dynamic program for a POMDP and use point-based
methods~\cite{pomdp}. The point-based methods rely on concavity of the value
function, which \edit{was established in Theorem~\ref{thm:DP-fin}.}

% \begin{proposition} \label{prop:concave}
%   The value functions $\{V_t\}_{t=1}^T$ defined in Theorem~\ref{thm:DP-fin}
%   are concave. 
% \end{proposition}
% See Appendix~\ref{app:concave} for proof.

\section{Proof of Theorem~\ref{thm:iid}} \label{sec:iid}

\edit{We follow the standard approach and show that the proposed leakage rate is
  optimal by showing achievability and a converse. As a preliminary step, we
  first show that under assumption (A), the objective can be rewritten in a
  simpler but equivalent form. To show achievability, we show that the
  proposed optimal policy belongs to a class of policies that satisfies a
  certain invariance property. Using this property the multi-letter mutual
  information expression can be reduced into a single-letter expression. For
  the converse we provide two proofs: the first uses dynamic programming and
  the second uses purely probabilistic and information theoretic arguments.}

\subsection {Simplification of the dynamic program} \label{sec:IID}

\edit{Define
\begin{equation*}
  \theta_t(s) = \PR^{\mathbf f}(S_t = s \mid Y^{t-1} = y^{t-1}, A^{t-1} = a^{t-1}). 
\end{equation*}
Then, under Assumption~(A), we can simplify the belief state $\pi_t$ as
follows:
\begin{align*}
  \pi_t(x,s) &= \PR(X_t = x, S_t = s | Y^{t-1} = y^{t-1}, A^{t-1} = a^{t-1})
  \\
  &\stackrel{(a)}= \PR(X_t = x | S_t = s, Y^{t-1} = y^{t-1}, A^{t-1} =
  a^{t-1}) \\
  & \quad \times 
  \PR(S_t = s \mid Y^{t-1} = y^{t-1}, A^{t-1} = a^{t-1})
  \\
  &\stackrel{(b)}= P_X(x) \theta_t(s)
\end{align*}
where $(a)$ follows from the product rule of probability and $(b)$ uses
Assumption~(A) and the definition of $\theta_t$.}

\edit{Since $\pi_t(x,s) = P_X(x) \theta_t(s)$,}
in principle, we can simplify the dynamic program of
Theorem~\ref{thm:DP-inf} by using $\theta_t$ as an information state. However, for
reasons that will become apparent, we provide an alternative simplification
that uses an information state $\xi_t \in \ALPHABET P_W$. 

Recall that $W_t = S_t - X_t$ which  takes values in $\ALPHABET W = \{ s -
x : s \in \ALPHABET S, x \in \ALPHABET X \}$. For any realization
$(y^{t-1},a^{t-1})$ of past observations and actions, define $\xi_t \in
\ALPHABET P_W$ as follows: for any $w \in \ALPHABET W$, 
\[
  \xi_t(w) = \PR^{\mathbf f}(W_t = w \mid Y^{t-1} = y^{t-1}, A^{t-1} = a^{t-1}).
\]
If $Y^{t-1}$ and $A^{t-1}$ are random variables, then $\xi_t$ is a $\ALPHABET
P_W$-valued random variable. 
As was the case for $\pi_t$, it can be shown that $\xi_t$ does not depend on
the choice of the policy~$\mathbf f$. 

\begin{lemma}\label{lem:equiv}
  Under \edit{Assumption}~(A), $\theta_t$ and $\xi_t$ are related as follows:
\begin{equation}\xi_t(w) = \sum_{(x,s) \in \ALPHABET D(w)} P_X(x) \theta_t(s).\label{eq:xi_theta}\end{equation}
\iffalse
  \begin{enumerate}
    \item $\xi_t(w) = \sum_{(x,s) \in \ALPHABET D(w)} P_X(x) \theta_t(s)$. 
    \item $\theta_t = P_X * \xi_t$.
  \end{enumerate}
\fi
\end{lemma}

\begin{proof}
  \begin{align*}
    \xi_t(w) &= \PR^{\mathbf f}(W_t = w \mid Y^{t-1} = y^{t-1}, A^{t-1} = a^{t-1}) \\
    &= \PR^{\mathbf f}(S_t - X_t = w \mid Y^{t-1} = y^{t-1}, A^{t-1} = a^{t-1}) \\
    &= \sum_{(x,s) \in \ALPHABET D(w)} P_X(x) \theta_t(s).
  \end{align*}
\end{proof}
\iffalse
  For part 2):
  \begin{align*}
    \theta_t(s) &= \PR^{\mathbf f}(S_t = s_t | Y^{t-1} = y^{t-1}, A^{t-1} = a^{t-1}) \\
    &= P(W_t + X_t = s_t | y^{t-1}, a^{t-1}) \\
    &= (P_X * \xi_t)(s_t).  \tag*{\qedhere}
  \end{align*}
\end{proof}
\fi
Since $\pi_t(x,s) = P_X(x)\theta_t(s)$, \edit{Lemma~\ref{lem:equiv} shows that
$\xi_t$ is a function of $\pi_t$. We will show that we can}
simplify the dynamic program of Theorem~\ref{thm:DP-inf} by using
$\xi_t$ as the information state instead of $\pi_t$. For such a simplification
to work, we would have to use charging policies of the form $q_t(y_t|w_t,
y^{t-1})$. We establish that restricting attention to such policies is without
loss of optimality. For that matter, define $\ALPHABET B$ as follows:
\begin{equation} \label{eq:B}
  \ALPHABET B = \left\{ b \in \ALPHABET P_{Y|W} : 
b(\ALPHABET Y_\circ(w) \mid w) = 1,\ \forall w \in \ALPHABET W \right\}.
\end{equation}

\begin{lemma} \label{lem:iid}
  Given $a \in \ALPHABET A$ and $\pi \in \mathcal P_{X,S}$, 
  define the following:
  \begin{itemize}
    \item $\xi \in \ALPHABET P_W$ as $\xi(w) = \sum_{(x,s) \in \ALPHABET
      D(w)} \pi(x,s)$ 

    \item $b \in \mathcal B$ as follows: for
      all $y \in \ALPHABET Y, w \in \ALPHABET W$
        \[
          b(y \mid w) = 
          \frac 
          { 
            \sum_{(x,s) \in \ALPHABET D(w)} 
            a(y \mid x, s) \pi( x, s) 
          }
          { 
            \xi(w)
          }; 
        \]
    \item $\tilde a \in \mathcal A$ as follows: for all $y \in \mathcal Y, x
      \in \mathcal X, s \in \mathcal S$
      \[
        \tilde a(y|x,s) = b(y|s-x).
      \]
  \end{itemize}
  Then under \edit{Assumption}~(A), we have
  \begin{enumerate}
    \item Invariant Transitions: for any $y \in \ALPHABET Y$, 
      $\varphi(\pi,y,a) = \varphi(\pi,y,\tilde a)$.
    \item Lower Cost: $I(a; \pi) \ge I(\tilde a; \pi) = I(b ; \xi)$. 
  \end{enumerate}
  Therefore, in the sequential problem of Sec.~\ref{sec:equiv}, there is no
  loss of optimality in restricting attention to actions $b \in \ALPHABET B$.
\end{lemma}
\begin{proof}
  \begin{enumerate}
    \item Suppose $(X,S) \sim \pi$ and $W = S - X$, $S_{+} = W + Y$, $X_{+}
      \sim P_X$. We will compare $\PR(S_{+} | Y)$ when $Y \sim a(\cdot |
      X,S)$ with when $Y \sim \tilde a(\cdot | X,S)$. Given $w \in \ALPHABET
      W$ and $y \in \ALPHABET Y$, 
      \begin{align}
        \PR^a(W = w, Y = y)  &= \sum_{(x,s) \in \ALPHABET D(w)}
        a(y|x,s) \pi(x,s) \notag \\
        &=  \sum_{(x,s) \in \ALPHABET D(w)} b(y|w) \pi(x,s) \notag \\
        &\stackrel{\text{(a)}}=
         \sum_{(x,s) \in \ALPHABET D(w)} \tilde a(y|x,s) \pi(x,s) \notag \\
        &= \PR^{\tilde a}(W = w, Y = y) \label{eq:WY}
      \end{align}
      where (a)~uses that for all $(x,s) \in \ALPHABET D(w)$, $s-x = w$.
      Marginalizing~\eqref{eq:WY} over $W$, we get that $\PR^a(Y = y) =
      \PR^{\tilde a}(Y = y)$. Since $S_{+} = W + Y$, Eq.~\eqref{eq:WY} also
      implies $\PR^{a}(S_{+} = s, Y = y) = \PR^{\tilde a}(S_{+} = s, Y = y)$. 
      Therefore, $\PR^{a}(S_{+} = s | Y = y) = \PR^{\tilde a}(S_{+} = s | Y =
      y)$.

    \item Let $(X,S) \sim \pi$ and $W = S - X$. Then $W \sim \xi$. Therefore,
      we have
      \[
        I(a;\pi) = I^a(X,S; Y) \ge I^a(W;Y).
      \]
      where the last inequality is the data-processing inequality. Under
      $\tilde a$, $(X,S) - W - Y$, therefore,
      \[
        I(\tilde a;\pi) = I^{\tilde a}(X,S; Y) = I^{\tilde a}(W;Y).
      \]

      Now, by construction, the joint distribution of $(W,Y)$ is the same
      under $a$, $\tilde a$, and $b$. Thus, 
      \[
        I^a(W;Y) = I^{\tilde a}(W;Y) = I^b(W;Y).
      \]
      Note that $I^b(W;Y)$ can also be written as $I(b;\xi)$. The result
      follows by combining all the above relations.
      \qedhere
  \end{enumerate}
\end{proof}

Once attention is restricted to actions $b \in \ALPHABET B$, the update of
$\xi_t$ may be expressed in terms of $b \in \ALPHABET B$ as follows:
\begin{lemma}\label{lem:xi}
  For any realization $y_t$ of $Y_t$ and $b_t$ of $B_t$, $\xi_{t+1}$ is given
  by
  \begin{equation} \label{eq:xi}
    \xi_{t+1} = \tilde \varphi(\xi_t, y_t, b_t)
  \end{equation}
  where $\tilde \varphi$ is given by
  \begin{multline*}
      \tilde{\varphi}(\xi,y,b)(w_{+}) \\
      = 
      \frac{ \sum_{x \in \mathcal X, w \in \mathcal W} P_X(x) \IND_{w_+}\{y + w - x\} b(y \mid w)
    \xi(w) } { \sum_{w \in \mathcal W} b(y \mid w) \xi(w)
  }.
  \end{multline*}
\end{lemma}
\begin{proof}
  \edit{The proof is similar to the proof of Lemma~\ref{lem:pi}.}
\end{proof}

For any $b \in \ALPHABET B$ and $\xi \in \ALPHABET P_W$, let us define the
Bellman operator $\tilde {\mathscr B}_b : [ \ALPHABET P_W \to \reals] \to [
\ALPHABET P_W \to \reals]$ as follows: \edit{for any $\tilde V \colon
\ALPHABET P_W \to \reals$ and any $\xi \in \ALPHABET P_W$,}
\begin{multline*}
  \tBELL{b}{\tilde V}{\xi} = I(b; \xi)\ + 
  \sum_{\substack{ y \in \ALPHABET Y, w \in \ALPHABET W }}
  \xi(w)b(y\mid w) \tilde V\big( \tilde \varphi(\xi, y, b) \big).
\end{multline*}

\begin{theorem} \label{thm:iid-DP}
   Under assumption~{(A)}, there is no loss of optimality in
  restricting attention to optimal policies of the form $q_t(y_t | w_t,
  \xi_t)$ in Problem~\ref{prob:original}. 
  \begin{enumerate}
    \item For the finite horizon case, we can identify the optimal policy
      $\mathbf q^* = (q^*_1, \dots, q^*_T)$ by iteratively defining 
      \emph{value functions} $\tilde V_t \colon \ALPHABET P_{W} \to \reals$.
      For any $\xi \in \ALPHABET P_{W}$, $\tilde V_{T+1}(\xi) = 0$, and for
      $t = T, T-1, \dots, 1$,
      \begin{equation} \label{eq:fin-DP-iid}
        \tilde V_t(\xi) = \min_{b \in \ALPHABET B}
        [ \tilde {\mathscr B}_b \tilde V_{t+1}](\xi).
      \end{equation}
      Let $f^\circ_t(\xi)$ denote the arg min of the right hand side
      of~\eqref{eq:fin-DP-iid}. Then, the optimal policy $\mathbf q^* =
      (q^*_1, \dots, q^*_T)$ is given by
      \[
        q^*_t(y_t | w_t, \xi_t) = b_t(y_t | w_t), 
        \text{ where } b_t = f^\circ_t(\xi_t).
      \]
      Moreover, the optimal (finite horizon) leakage rate is given
      by $\tilde V_1(\xi_1)/T$, where $\xi_1(w) = \sum_{(x,s) \in \ALPHABET
      D(w)} P_X(x) P_{S_1}(s)$. 

    \item For the infinite horizon, \edit{suppose that there exists a 
      constant $\tilde J \in \reals$ and a function $\tilde v :
    \ALPHABET P_{S} \to \reals$ which satisfies} 
      the following fixed point equation:
      \begin{equation} \label{eq:inf-DP-iid}
        \tilde J + \tilde v(\xi) = \min_{b \in \ALPHABET B}
        [\tilde {\mathscr B}_b \tilde v](\xi),\ \forall \xi \in \ALPHABET P_{W}.
      \end{equation}
      Let $\mathbf f^\circ(\xi)$ denote the arg min of the right hand side
      of~\eqref{eq:inf-DP-iid}. Then, the time-homogeneous policy $\mathbf
      q^* = (q^*, q^*, \dots)$ given by
      \[
        q^*(y_t | w_t, \xi_t) = b_t(y_t | w_t), 
        \text{ where } b_t = f^\circ(\xi_t)
      \]
      is optimal. Moreover, the optimal (infinite horizon) leakage rate is
      given by $\tilde J$.
      \qed
  \end{enumerate}
\end{theorem}
\begin{proof}
  Lemma~\ref{lem:xi} implies that $\{\xi_t\}_{t \ge 1}$ is a controlled
  Markov process with control action $b_t$. Lemma~\ref{lem:iid}, part~2),
  implies that the per-step cost can be written as 
  \[
    \frac 1T \EXP\bigg[ \sum_{t=1}^T I(b_t;\xi_t) \bigg].
  \]
  Thus, by standard results in Markov decision
  \edit{theory}~\cite{hl:MDP}, the
  optimal solution is given by the dynamic program described above.
\end{proof}

\subsection{Weak achievability} \label{sec:achievability}

To simplify the analysis, we assume that we are free to choose the initial
distribution of the state of the battery, which could be done by, for example,
initially charging the battery to a random value according to that
distribution. In principle, such an assumption could lead to a lower
achievable leakage rate. For this reason, we call it \emph{weak}
achievability. In the next section, we will show achievability starting from
an arbitrary initial distribution, which we call \emph{strong} achievability.

\iffalse
\begin{definition}
  Any $\theta \in \ALPHABET P_S$ and $\xi \in \ALPHABET P_W$ are said to \edit{be}
  \emph{equivalent} to each other if they satisfy the transformation in
  Lemma~\ref{lem:equiv}.
\end{definition}
\fi

\begin{definition}[Constant-distribution policy]
  A time-homogeneous policy $\mathbf f^\circ = (f^\circ, f^\circ, \dots)$ is 
\edit{  called} a constant-distribution policy if for all $\xi \in \ALPHABET P_W$,
  $f^\circ(\xi)$ is a constant. If $f^\circ(\xi) = b^\circ$, then with a
  slight abuse of notation, we refer to $\mathbf b^\circ = (b^\circ, b^\circ,
  \dots)$ as a constant-distribution policy.
\end{definition}

Recall that under a constant-distribution policy $b \in \ALPHABET B$,  for any
realization $y^t$ of $Y^t$, $\theta_t$ and $\xi_t$ are given as follows:
\edit{\begin{align*}
  \theta_t(s) &= \PR(S_t = s \mid Y^{t-1} = y^{t-1}, B^{t-1} = b^{t-1}) \\
  \xi_t(w) &= \PR(W_t = w \mid Y^{t-1} = y^{t-1}, B^{t-1} = b^{t-1}).
\end{align*}}

\subsubsection{Invariant Policies}

We next impose an invariance property on the class of policies.
Under this restriction the leakage rate expression will simplify substantially.
Subsequently we will show that the optimal policy belongs to this restricted class.

\begin{definition}[Invariance Property]
  For a given distribution $\theta_1$ of the initial battery state, a
  constant-distribution policy $b \in \ALPHABET B$ is called \edit{an} \emph{invariant} policy if for all~$t$, $\theta_t = \theta_1$ and $\xi_t = \xi_1$,
  \edit{where $\xi_1$ and $\theta_1$ satisfy~\eqref{eq:xi_theta}}.
\label{def:inv}
\end{definition}

\begin{remark}\label{rem:iid}
  An immediate implication of the above definition is that
  under any invariant policy~$\mathbf b$, the conditional distribution
  $\PR^{\mathbf b}(X_t, S_t, Y_t | Y^{t-1})$ is the same as the joint
  distribution $\PR^{\mathbf b}(X_1, S_1, Y_1)$. Marginalizing over $(X,S)$ we
  get that $\{Y_t\}_{t \ge 1}$ is an i.i.d.\ sequence.
\end{remark}

\begin{lemma} \label{lem:structured}
  If the system starts with an initial distribution $\theta$ of the battery
  state, and \edit{$\xi$ and $\theta$ satisfy~\eqref{eq:xi_theta}}, then an invariant policy
  $\mathbf b = (b, b, \dots)$ corresponding to $(\theta, \xi)$ achieves a
  leakage rate 
  \[
    L_T(\mathbf b) = I(W_1; Y_1) = I(b;\xi)
  \]
  for any horizon~$T$.
\end{lemma}

\begin{proof}
The proof is a simple corollary of the invariance property in Definition~\ref{def:inv}.
Recall from the dynamic program of Theorem~\ref{thm:iid-DP}, that the
  performance of any policy $\mathbf b = (b_1, b_2, \dots)$ such that $b_t \in
  \ALPHABET B$, is given by
  \[
    L_T(\mathbf b) = \frac 1T \EXP\bigg[ \sum_{t=1}^T I(b_t; \xi_t) \bigg].
  \]
  Now, we start with an initial distribution $\xi_1 = \xi$ and follow the
  constant-distribution structured policy $\mathbf b = (b, b, \dots)$. Therefore,
  by \edit{Definition~\ref{def:inv}}, $\xi_t = \xi$ for all~$t$. Hence, the
  leakage rate under policy $\mathbf b$ is 
  \begin{equation}
    L_T(\mathbf b) = I(b; \xi).
    \tag*{\qedhere}
  \end{equation}

\end{proof}

\begin{remark}
Note that Lemma~\ref{lem:structured} can be easily derived independently of Theorem~\ref{thm:iid-DP}. 
From  Proposition~\ref{prop:A-B}, we have that:
  \begin{equation} \label{eq:I-b-1}
    L_T(\mathbf b) = \frac 1T \sum_{t=1}^T I^{\mathbf b}(S_t, X_t; Y_t |
    Y^{t-1}).
  \end{equation}
  From Remark~\ref{rem:iid}, we have that $I^{\mathbf b}(S_t,
  X_t; Y_t | Y^{t-1}) = I^{b}(S_1, X_1; Y_1) = I^{b}(W_1; Y_1)$, which immediately results in
Lemma~\ref{lem:structured}.
\end{remark}

For invariant policies we can express the leakage rate in the following fashion, which is useful in the proof of optimality.

\begin{lemma} \label{lem:MI}
  For any invariant policy $b$, 
  \[ I^{b}(W_1; Y_1) = I^{b}(W_1; X_1). \]
\end{lemma}
\begin{proof}
  Consider the following sequence of simplifications:
  \begin{align*}
    I^b(W_1; Y_1) 
    &= H^b(W_1) - H^b(W_1 | Y_1) \\
    &= H^b(W_1) - H^b(W_1 + Y_1 | Y_1) \\
    &\stackrel{\text{(a)}}= H^b(W_1) - H^b(S_2 | Y_1) \\
    &\stackrel{\text{(b)}}= H^b(W_1) - H^b(S_1) \\
    &\stackrel{\text{(c)}}= H^b(W_1) - H^b(S_1 | X_1) \\
    &\stackrel{\text{(d)}}= H^b(W_1) - H^b(W_1 | X_1) \\
    &= I^b(W_1; X_1).
  \end{align*}
  where~(a) is due to the battery update equation~\eqref{eq:conservation};
  (b)~is because $b$ is an invariant ; (c)~is because $S_1$ and
  $X_1$ are independent; and (d)~is because $S_1 = W_1 + X_1$.
\end{proof}

\subsubsection{Structured Policy}

We now introduce a class of policies that satisfy the invariance property in Def.~\ref{def:inv}.
This will be then used in the proof of Theorem~\ref{thm:iid}.
\begin{definition}[Structured Policy] \label{def:structured}
  Given $\theta \in \ALPHABET P_S$ and $\xi \in \ALPHABET P_W$, a
  constant-distribution policy $\mathbf b = (b, b, \dots)$ is called a
  \emph{structured policy} with respect to $(\theta, \xi)$ if:
  \[
    b(y|w) = \begin{cases}
      P_X(y) \frac{\theta(y+w)}{\xi(w)},
      & y \in \ALPHABET X \cap \ALPHABET Y_\circ(w) \\
      0, & \text{otherwise}.
    \end{cases}
  \]
\end{definition}

Note that it is easy to verify that the distribution~$b$ defined above is a
valid conditional probability distribution. 

\begin{lemma}\label{lem:structure-invariance}
  For any $\theta \in \ALPHABET P_S$ and $\xi \in \ALPHABET P_W$, the
  structured policy $\mathbf b = (b,b,\dots)$ given in
  Def.~\ref{def:structured} is an invariant  policy.
\end{lemma}
\begin{proof}
 \edit{ Since $(\theta_t, \xi_t)$ are related according to Lemma~\ref{lem:equiv}}, in
  order to check whether a policy is invariant it is sufficient to check
  that $\theta_t = \theta_1$ for all~$t$. Furthermore, to check if a
  \emph{time-homogeneous} policy is an invariant policy, it is sufficient to
  check that either $\theta_2 = \theta_1$. 

  Let the initial distributions $(\theta_1, \xi_1) = (\theta, \xi)$ and the
  system variables be defined as usual. Now consider a realization $s_2$ of
  $S_2$ and $y_1$ of $Y_1$. This means that $w_1 = s_2 - y_1$. Since $Y_1$ is
  chosen according to \edit{$b(\cdot|w_1)$}, it must be that $y_1 \in \ALPHABET
  X \cap \ALPHABET Y_\circ(w_1)$. Therefore, 
  \begin{align}
    \PR^{\mathbf b}(S_2 = s_2, Y_1 = y_1) &= 
    \PR^{\mathbf b}(S_2 = s_2, Y_1 = y_1, W_1 = s_2 - y_1) \notag \\
    &= \xi_1(s_2 - y_1) b(y_1 | s_2 - y_1) \notag \\
    &= P_X(y_1) \theta_1(s_2),
  \end{align}
  where in the last equality we use the fact that $y_1 \in \ALPHABET X \cap
  \ALPHABET Y_\circ(s_2 - y_1)$.  Note that if $y_1 \not\in \ALPHABET X \cap
  \ALPHABET Y_\circ(s_2 - y_1)$, then $\PR^{\mathbf b}(S_2 = s_2, Y_1 = y_1) =
  0$.  Marginalizing over $s_2$, we get $\PR^{\mathbf b}(Y_1 = y_1) = P_X(y_1)$. 
  
  Consequently, $\theta_2(s_2) = \PR^{\mathbf b}(S_2 = s_2 | Y_1 = y_1) =
  \theta_1(s_2)$. Hence, $\mathbf b$ is invariant as required. 
\end{proof}

\begin{remark}\label{rem:PY=PX}
  As argued in Remark~\ref{rem:iid}, under any invariant policy, $\{Y_t\}_{t
  \ge 1}$ is an i.i.d.\ sequence. As argued in the proof of
  Lemma~\ref{lem:structure-invariance}, for a structured policy the marginal
  distribution of $Y_t$ is $P_X$. Thus, an eavesdropper cannot statistically
  distinguish between $\{X_t\}_{t \ge 1}$ and $\{Y_t\}_{t \ge 1}$.
\end{remark}

\begin{proposition}
  Let $\theta^*$, $\xi^*$, and $b^*$ be as defined in Theorem~\ref{thm:iid}.
  Then,
  \begin{enumerate}
    \item \edit{$(\theta^*, \xi^*)$ satisfy~\eqref{eq:xi_theta}};
    \item $b^*$ is a structured policy with respect to $(\theta^*, \xi^*)$.
    \item If the system starts in the initial battery state $\theta^*$ and
      follows the constant-distribution policy $\mathbf b^* = (b^*,
      b^*,\dots)$, the leakage rate  is given by $J^*$.
  \end{enumerate}
  Thus, the performance $J^*$ is achievable. 
\end{proposition}
\begin{proof}
  The proofs of parts 1) and 2) follows from the definitions. The proof of
  part~3) follows from Lemmas~\ref{lem:structured} and~\ref{lem:MI}.
\end{proof}

This completes the proof of the achievability of Theorem~\ref{thm:iid}.

\subsection{Strong achievability}
\label{sec:strong_achievability}
\begin{lemma}\label{lem:converge}
  Assume that for any $x \in \ALPHABET X$, $P_X(x) > 0$. 
  Let $(\theta^\circ, \xi^\circ)$ \edit{be a pair satisfying~\eqref{eq:xi_theta}} and $\mathbf b^\circ =
  (b^\circ, b^\circ, \dots)$ be the corresponding structured policy. 
  
  Assume that $\theta^\circ \in \mathit{int}(\ALPHABET P_S)$ or equivalently,
  for any $w \in \ALPHABET W$ and $y \in \ALPHABET X \cap \ALPHABET
  Y_\circ(w)$, $b^\circ(y|w) > 0$. Suppose the system starts in the initial
  state $(\theta_1, \xi_1)$ and follows policy $\mathbf b^\circ$.
  Then:
  \begin{enumerate}
    \item the process $\{\theta_t\}_{\ge 1}$ converges weakly to
      $\theta^\circ$;
    \item the process $\{\xi_t\}_{\ge 1}$ converges weakly to $\xi^\circ$; 
    \item for any continuous function $c \colon \ALPHABET P_W
      \to \reals$, 
      \begin{equation}\label{eq:limit}
        \lim_{T \to \infty} \frac 1T \sum_{t=1}^T E[c(\xi_t)] = c(\xi^\circ).
      \end{equation}
    \item Consequently, the infinite horizon leakage rate under $\mathbf
      b^{\circ}$~is 
      \[
        L_\infty(\mathbf b^{\circ}) = I(b^\circ, \xi^\circ).
      \]
  \end{enumerate}
\end{lemma}
\begin{proof}
  The proof of parts 1) and 2) is presented in Appendix~\ref{app:converge}.
  From~2), $\lim_{t \to \infty} \EXP[ c(\xi_t) ] = c(\xi^\circ)$, which
  implies~\eqref{eq:limit}. Part~4) follows from part~3) by setting $c(\xi_t)
  = I(b^\circ, \xi_t)$.
\end{proof}

Proposition~\ref{prop:convex} implies that $\theta^*$ defined in
Theorem~\ref{thm:iid} lies in $\mathit{int}(\ALPHABET P_S)$. Then, by
Lemma~\ref{lem:converge}, the constant-distribution policy $\mathbf b^* = (b^*,
b^*, \dots)$ (where $b^*$ is given by Theorem~\ref{thm:iid}), achieves the
leakage rate $I(b^*,\xi^*)$. By Lemma~\ref{lem:MI}, $I(b^*,\xi^*)$ is same as
$J^*$ defined in Theorem~\ref{thm:iid}. Thus, $J^*$ is achievable starting
from any initial state $(\theta_1, \xi_1)$. 

\subsection{Dynamic programming converse}

We provide two converses. One is based on the dynamic program of
Theorem~\ref{thm:iid-DP}, which is presented in this section; the other is
based purely on information theoretic arguments, which is presented in the
next section. 

In the dynamic programming converse, we show that for $J^*$ given in
Theorem~\ref{thm:iid}, $v^*(\xi) = H(\xi)$, and any $b \in \ALPHABET B$, 
\begin{equation} \label{eq:ACOI}
  J^* + v^*(\xi) \le \tBELL{b}{v^*}{\xi}, \quad \forall \xi \in 
\ALPHABET P_W,
\end{equation}
Since $H(\xi)$ is bounded, from~\cite[Lemma~5.2.5(b)]{hl:MDP}, we get that 
$J^*$ is a lower bound of the optimal leakage rate.

To prove~\eqref{eq:ACOI}, pick any $\xi \in \ALPHABET P_W$ and $b \in
\ALPHABET B$. Suppose $W_1 \sim \xi$, $Y_1 \sim b(\cdot | W_1)$, $S_2 = Y_1 +
W_1$, $X_2$ is independent of $W_1$ and $X_2 \sim P_X$ and $W_2 = S_2 - X_2$.
Then,
\begin{align}
  \tBELL{b}{v^*}{\xi} &= I(b;\xi) +
  \smashoperator{\sum_{(w_1, y_1) \in \ALPHABET W \times \ALPHABET Y}}
  \xi(w_1) b(y_1|w_1) v^*(\hat \varphi(\xi, y_1, b)) \notag \\
  &= I(W_1;Y_1) + H(W_2|Y_1)
  \label{eq:DPc-1}
\end{align}
where the second equality is due to the definition of conditional entropy. 
Consequently,
\begin{align} 
  \hskip 2em & \hskip -2em
  \tBELL{b}{v^*}{\xi} - v^*(\xi) = H(W_2|Y_1) - H(W_1|Y_1) \notag \\
  &= H(W_2 | Y_1) - H(W_1 + Y_1 | Y_1) \notag \\
  &\stackrel{\text{(a)}}= H(S_2 - X_2 | Y_1) - H(S_2 | Y_1)\notag \\
  &\stackrel{\text{(b)}}\ge \min_{\theta_2 \in \ALPHABET P_S} 
    \big[ H(\tilde S_2 - X_2) - H(\tilde S_2) \big],
    \qquad \tilde S_2 \sim \theta_2 
    \notag \\
  &= J^*
  \label{eq:DPc-2}
\end{align}
where (a)~uses $S_2 = Y_1 + W_1$ and $W_2 = S_2 - X_2$; (b)~uses the fact
that $H(A_1 | B) - H(A_1 - A_2 |B) \ge \min_{P_{A_1}} \big[ H(A_1) - H(A_1 -
A_2) \big]$ for any joint distribution on $(A_1, A_2, B)$.

The equality in~\eqref{eq:DPc-2} occurs when $b$ is an invariant policy
and $\theta_2$ is same as $\theta^*$ defined in Theorem~\ref{thm:iid}. For
\edit{$\xi$ that are not equivalent to $\theta^*$ via~\eqref{eq:xi_theta}}, the inequality
in~\eqref{eq:DPc-2} is strict. 

We have shown that Eq.~\eqref{eq:ACOI} is true. Consequently, $J^*$ is a
lower bound on the optimal leakage rate $\tilde J$.

\subsection{Information theoretic converse}
\label{subsec:IT_Conv}

Consider the following inequalities: for any admissible policy $\mathbf q \in
\mathcal Q_B$, we have
\begin{align}
  I(S_1,X^T;Y^T) 
  &= \sum_{t = 1}^{T} I(S_t,X_t;  Y_t | Y^{t-1})  \notag \\
  &\stackrel{\text{(a)}}\geq 
  \sum_{t = 1}^{T} I(W_t;  Y_t | Y^{t-1})  \label{eq:converse-1}
\end{align}
where (a) follows from \edit{the fact that $W_t = X_t- S_t$ is a deterministic function of $(X_t, S_t)$ and that the mutual information is non-negative.}

Now consider
\begin{align}
  \hskip 2em & \hskip -2em
  I(W_t; Y_t | Y^{t-1}) = H(W_t|Y^{t-1}) - H(W_t|Y^t) \notag \\
  &= H(W_t|Y^{t-1}) - H(W_t + Y_t | Y^t) \notag \\
  &\stackrel{\text{(b)}}= H(W_t|Y^{t-1}) - H(S_{t+1} | Y^t) \notag \\
  &\stackrel{\text{(c)}}= H(W_t|Y^{t-1}) - H(S_{t+1} | Y^t, X_{t+1}) \notag \\
  &= H(W_t|Y^{t-1}) - H(S_{t+1} - X_{t+1} | Y^t, X_{t+1}) \notag \\
  &\stackrel{\text{(d)}}= H(W_t|Y^{t-1}) - H(W_{t+1} | Y^t, X_{t+1}) 
  \label{eq:converse-2}
\end{align}
where (b)~follows from~\eqref{eq:conservation}; (c)~follows because of
assumption~(A); and (d)~also follows from~\eqref{eq:conservation}.

Substituting~\eqref{eq:converse-2} in~\eqref{eq:converse-1} (but expanding
the last term as $H(W_T|Y^{T-1}) - H(W_T |Y^T)$, we get
\begin{align}
  \hskip 1em & \hskip -1em
  I(S_1, X^T; Y^T) \ge \sum_{t = 1}^{T} I(W_t;  Y_t | Y^{t-1})  \\
&=  \sum_{t=1}^{T-1} \big[ H(W_t|Y^{t-1}) - H(W_{t+1} | Y^t, X_{t+1}) \big] \notag \\
     &\qquad \edit{+ H(W_T|Y^{T-1}) - H(W_T |Y^T) } \label{eq:indices_2_T_1} \\
    &= H(W_1) + \sum_{\edit{t=2}}^{\edit{T}} \big[  - H(W_t | Y^{t-1}, X_t) + H(W_t |
  Y^{t-1}) \big] \notag \\
  &\qquad - H(W_T | Y^T)  \label{eq:indices_2_T} \\
  &= H(W_1) + \sum_{t=2}^{T} I(W_t ; X_t | Y^{t-1}) - H(W_T | Y^T).
  \label{eq:converse-3}
\end{align}

Now, we take the limit $T \to \infty$ to obtain a lower bound to the leakage
rate: 
\begin{align*}  
  &L_\infty(\mathbf q) = \limsup_{T \rightarrow \infty}   
  \frac 1T I(S_1, X^T; Y^T) \\
  &\geq \limsup_{T \rightarrow \infty}   
  \frac 1T \bigg[ H(W_1) + \sum_{t=2}^{T} I(W_t ; X_t | Y^{t-1}) - H(W_T |
  Y^T)\bigg] \\
  &\stackrel{\text{(a)}}= {\edit{\limsup}_{T \rightarrow \infty}}   
  \frac 1T \left[ \sum_{t=2}^{T} I(W_t ; X_t | Y^{t-1}) \right] \\
  &\stackrel{\text{(b)}}\geq \min_{P_S \in \mathcal P_S} I(S - X;X) = J^*
\end{align*}
where (a) is because the entropy of any discrete random variable is bounded
\edit{and (b) follows from the fact that each term in the summation satisfies:
\begin{equation}
\label{eq:b_lb}
I(W_t ; X_t | Y^{t-1}) \ge \min_{P_S \in \mathcal P_S} I(S - X;X),
\end{equation}
which can be justified as follows. First note that for any policy ${\mathbf q}$ we have that:
$$I(W_t ; X_t | Y^{t-1}) = I(S_t- X_t; X_t | Y^{t-1})$$
depends on the joint distribution $P^{\mathbf q}_{S_t, X_t, Y^{t-1}}(s_t,x_t,y^{t-1})$ which  factors as:
\begin{align}
P^{\mathbf q}_{S_t, X_t, Y^{t-1}}(s_t,x_t,y^{t-1}) = P_X(x_t) P^{\mathbf q}_{S_t, Y^{t-1}}(s_t, y^{t-1})  \label{eq:P_Fac}
\end{align}
as $X_t$ is sampled i.i.d.\ from the distribution $P_X(\cdot)$ and from the state update equation~\eqref{eq:conservation}, we have that 
$S_t$ is a function of $(X^{t-1}, Y^{t-1})$, and thus independent of $X_t$. Now note that:
{\allowdisplaybreaks{\begin{align}
& I(W_t ; X_t | Y^{t-1}) \notag \\ &= \sum_{{y^{t-1} \in {\ALPHABET Y}^{t-1} }} I(S_t-X_t ; X_t | Y^{t-1} = y^{t-1})p(y^{t-1})\notag\\
&\ge \min_{y^{t-1} \in {\ALPHABET Y}^{t-1}} I(S_t-X_t ; X_t | Y^{t-1} = y^{t-1}) \notag \\
&\ge \min_{y^{t-1} \in {\ALPHABET Y}^{t-1}}  \min_{p_{S_t, X_t | Y^{t-1}}(\cdot| y^{t-1})} I(S_t-X_t ; X_t | Y^{t-1} = y^{t-1}) \notag\\
&\ge \min_{y^{t-1} \in {\ALPHABET Y}^{t-1}, {p_{S_t | Y^{t-1}}(\cdot| y^{t-1})}} I(S_t-X_t ; X_t | Y^{t-1} = y^{t-1}) \label{eq:x_is_ind}
\end{align}}}where~\eqref{eq:x_is_ind} follows from~\eqref{eq:P_Fac} so that $X_t$ is independent of $(S_t, Y^{t-1})$ and distributed according to $P_X(\cdot)$
Since~\eqref{eq:x_is_ind} is simply equivalent to minimizing over the distribution $P_S(\cdot)$, the relation in~\eqref{eq:b_lb} holds. This completes the proof.
 This shows that
$J^*$ is a lower bound to the minimum (infinite horizon) leakage rate.}

\section{Conclusions and Discussion}

In this paper, we study a smart metering system that uses a rechargeable
battery to partially obscure the user's power demand. Through a series of
reductions, we show that the problem of finding the best battery charging
policy can be recast as a Markov decision process. Consequently, the
optimal charging policies and the minimum information leakage rate are
given by the solution of an appropriate dynamic program.

For the case of i.i.d.\ demand, we provide an explicit characterization of the
optimal battery policy and the leakage rate.  In this special case it
suffices to choose a memoryless policy where the distribution of $Y_t$
depends only on $W_t$. Our achievability results rely on restricting
attention to a class of invariant policies. Under an invariant policy, the
consumption $\{Y_t\}_{t \ge 1}$ is i.i.d.\ and the leakage rate is
characterized by a single-letter mutual information expression. We then
further restrict attention to what we call structured policies under which
the marginal distribution of $\{Y_t\}_{t \ge 1}$ is~$P_X$. Thus, under the
structured policies, an eavesdropper cannot statistically distinguish between
$\{X_t\}_{t \ge 1}$ and $\{Y_t\}_{t \ge 1}$. We provide two converses; one is
based on the dynamic programming argument while the other is based on a
purely information theoretic argument. 

Extending of our MDP formulation to incorporate an additive cost, such as
the price of consumption, is rather immediate. However, the approach
presented in this work for explicitly characterizing the optimal leakage rate
in the i.i.d.\ case may not immediately extend to such general cost
functions. In another direction one can allow for a certain controlled wastage of energy drawn from the grid
to increase privacy. It would be interesting to see how the leakage rate decreases with the wasted energy.
The study of such problems, as well as finer implementation details of the
proposed system, remains an interesting future direction.

\appendices

\section{Proof of Property~\ref{prop:convex}} \label{app:convex}

For any ${\theta \in \mathit{int}(\mathcal P_S)}$ and ${\delta(s): \mathcal S
\rightarrow \reals}$ such that ${\sum_{s \in \mathcal S} \delta(s) = 0}$.
Let $\theta_\alpha(s) := \theta(s) + \alpha \delta(s)$. Then for small enough
$\alpha$, $\theta_\alpha \in \mathcal P_S$. Given such a $\theta_\alpha$, let
$\PR_{W,X}(w,x) = \PR_{W|X}(w|x)P_X(x) = \theta_{\alpha}(w+x)P_X(x)$. Then to
show that $I(W; X)$ is strictly convex on $\mathcal P_S$ we require
$\frac{d^2 I(W;X)}{d \alpha^2} > 0$. \edit{Due to independence of $X$ and $S$}, $I(W;X) =
H(W) - H(S)$. Therefore,
\begin{align*} 
  &\frac{d I(W;X)}{d \alpha} = \frac{d \left[- H(S) + H(W) \right]}{d \alpha} \\
  &=\sum_{\tilde s} \delta(\tilde s)\ln \theta_\alpha(\tilde s)
  - \sum_{w \in \mathcal W, s \in \mathcal S} 
  P_X(s - w) \delta(s) \ln P_W(w)  \\
  &\frac{d^2 I(W;X)}{d \alpha^2} = \sum_{s}
  \frac{\delta(s)^2}{\theta_\alpha(s)}  -\sum_{w \in \mathcal W}
  \frac{\left(\sum_{\tilde s \in \mathcal S}P_X(\tilde s - w)\delta(\tilde
s)\right)^2 } {P_W(w)}. 
\end{align*}
Let ${a_w(s) = \delta(s) \sqrt{\frac{P_X(s-w)}{\theta_{\alpha}(s)}}}$
and ${b_w(s) = \sqrt{\theta_{\alpha(s)}P_X(s-w)}}$. Using the Cauchy-Schwarz inequality, we can show that
\begin{align*} 
  &\frac{d^2 I(W;X)}{d \alpha^2} = \sum_{s}
  \frac{\delta(s)^2}{\theta_\alpha(s)}  -\sum_{w \in \mathcal W}
  \frac{\left(\sum_{\tilde s \in \mathcal S} a_w(\tilde s) b_w(\tilde s) \right)^2} {P_W(w)} \\
  &> \sum_{s} \frac{\delta(s)^2} {\theta_\alpha(s)}  
  -\sum_{w \in \mathcal W}
  \frac{\left(\sum_{\tilde s \in \mathcal S} a_w(\tilde s)^2 \right)
  \left(\sum_{\hat s \in \mathcal S} b_w(\hat s)^2 \right)} {P_W(w)} \\
  &= \sum_{s} \frac{\delta(s)^2} {\theta_\alpha(s)}  
  -\sum_{w \in \mathcal W}
  \left(\sum_{\tilde s \in \mathcal S} a_w(\tilde s)^2 \right) = 0.
\end{align*}

The strict inequality is because $a$ and $b$ cannot be linearly dependent. To
see this, observe that  $\frac{a(s)}{b(s)} = \frac{\delta(s)}{\theta(s) +
\alpha \delta(s)}$ cannot be equal to a constant for all $s \in \mathcal S$
since $\delta$ must contain negative as well as positive elements.

\section{Proof of Proposition \ref{prop:A-B}} \label{app:A-B}

The proof of Proposition~\ref{prop:A-B} relies on the following intermediate
results (which are proved later):
\begin{lemsec} \label{lem:a1}  
  For any $\mathbf {q} \in Q_A$,
  \[
    I^{\mathbf q}(S_1, X^T; Y^T) \ge \sum_{t=1}^T I^{\mathbf q}(X_t, S_t ; Y_t | Y^{t-1})
  \] 
  with equality if and only if $q \in Q_B$. 
\end{lemsec}

\begin{lemsec} \label{lem:a2}
  For any $\mathbf q_a \in Q_A$, there exists a $\mathbf q_b \in Q_B$, such that
  \[
    \sum_{t=1}^T I^{\mathbf q_a}(X_t, S_t ; Y_t | Y^{t-1}) 
    =
    \sum_{t=1}^T I^{\mathbf q_b}(X_t, S_t ; Y_t | Y^{t-1}).
  \]  
\end{lemsec}

Combining Lemmas~\ref{lem:a1} and~\ref{lem:a2}, we get that for any 
$\mathbf q_a \in \mathcal Q_A$, there exists a $\mathbf q_b \in \mathcal Q_B$
such that
\[
  I^{\mathbf q_a}(S_1, X^T; Y^T) \ge I^{\mathbf q_b}(S_1, X^T; Y^T).
\]
Therefore there is no loss of optimality in restricting attention to charging
policies in $\mathcal Q_B$.  Furthermore, Lemma~\ref{lem:a1} shows that for
any $q \in \mathcal Q_B$, $L_T(q)$ takes the additive form as given in
the statement of the proposition. 

\begin{proof}[Proof of Lemma~\ref{lem:a1}]
  For any $\mathbf q \in \mathcal Q_A$, we have
  \begin{align*}
    I^{\mathbf q}(S_1, X^n ; Y^n) 
    &\stackrel{\text{(a)}}=
    \sum_{t=1}^n I^{\mathbf q}(S_1, X^t ; Y_t | Y^{t-1}) \\
    &\stackrel{\text{(b)}}=
    \sum_{t=1}^n I^{\mathbf q}(X^t, S^t ; Y_t | Y^{t-1}) \\
    &\stackrel{\text{(c)}}\geq 
    \sum_{t=1}^n I^{\mathbf q}(X_t, S_t ; Y_t | Y^{t-1})
  \end{align*}
  where (a) uses the chain rule of mutual information and the fact that 
  $(\edit{S^{t-2}}, Y^{t-1}) \rightarrow X_{t-1} \rightarrow X_t$;\footnote{The notation
    ${A \rightarrow B \rightarrow C}$ is used to indicate that $A$ is
  conditionally independent of $C$ given $B$.}
  (b) uses the fact that the battery process $S^t$ is a
  deterministic function of $S_1$, \edit{$X^{t-1}$}, and \edit{$Y^{t-1}$} given
  by~\eqref{eq:conservation}; and (c) uses the fact that removing terms
  \edit{does not reduce} the mutual information.
\end{proof}

\begin{proof} [Proof of Lemma~\ref{lem:a2}]
  For any ${\mathbf q_a = (q^a_1, q^a_2,\dots, q^a_T) \in \mathcal Q_A}$, 
  construct a ${\mathbf q_b  = (q^b_1, q^b_2,\dots, q^b_T) \in \mathcal
  Q_B}$ as follows: for any $t$ and realization $(x^t, s^t, y^t)$ of $(X^t,
  S^t, Y^t)$ let
  \begin{equation} \label{eq:qb-def}
    q^{b}_{t}(y_t | x_t, s_t, y^{t-1} ) 
    = 
    \PR^{\mathbf q_a}_{Y^{}_{t} | X_t, S_t, Y^{t-1}}(y_t | x_t, s_t, y^{t-1}). 
  \end{equation}

  To prove the Lemma, we show that for any~$t$, 
  \begin{equation} \label{eq:pr-eq}
    \PR^{\mathbf q_a}_{X_t, S_t, Y^{t}} 
    = \PR^{\mathbf q_b}_{X_t, S_t, Y^{t}}. 
  \end{equation}
  By definition of $\mathbf q_b$ given by~\eqref{eq:qb-def}, 
  to prove~\eqref{eq:pr-eq}, it is sufficient to show that 
  \begin{equation} \label{eq:pr-eq-2}
    \PR^{\mathbf q_a}_{X_t, S_t, Y^{t-1}} 
    = \PR^{\mathbf q_b}_{X_t, S_t, Y^{t-1}}.
  \end{equation}
  We do so using induction.

  For $t=1$, $\PR^{\mathbf q_a}_{X_1, S_1}(x,s) = P_{X_1}(x) P_{S_1}(s) =
  \PR^{\mathbf q_b}_{X_1, S_1}(x,s)$. This forms the basis of induction.
  Now assume that \eqref{eq:pr-eq-2} hold for~$t$. 
  
  In the rest of the proof, for ease of notation, we denote $\PR^{\mathbf q_a}_{X_{t+1}, S_{t+1},
  Y^{t}}(x_{t+1}, s_{t+1}, y^t)$ simply by $\PR^{\mathbf q_a}(x_{t+1},
  s_{t+1}, y^t)$. For $t+1$, we have
  \begin{align*}
    \hskip 1em & \hskip -1em
    \PR^{\mathbf q_a}(x_{t+1}, s_{t+1}, y^t) =
    \smashoperator{\sum_{(x_t, s_t) \in \ALPHABET X \times \ALPHABET S}}
    \PR^{\mathbf q_a}(x_{t+1}, x_t, s_{t+1}, s_t, y^t) \\
    &=
    \smashoperator{\sum_{(x_t, s_t) \in \ALPHABET X \times \ALPHABET S}}
    Q(x_{t+1}|x_t) \IND_{s_{t+1}}\{s_t - x_t + y_t \} 
    q_a(y_t| x_t, s_t, y^{t-1}) \\
    &\qquad \times \PR^{\mathbf q_a}(x_t, s_t, y^{t-1}) \\
    &\stackrel{\text{(a)}}=
    \smashoperator{\sum_{(x_t, s_t) \in \ALPHABET X \times \ALPHABET S}}
    Q(x_{t+1}|x_t) \IND_{s_{t+1}}\{s_t - x_t + y_t \} 
    q_b(y_t| x_t, s_t, y^{t-1}) \\
    &\qquad \times \PR^{\mathbf q_b}(x_t, s_t, y^{t-1}) \\
    &= 
    \PR^{\mathbf q_b}(x_{t+1}, s_{t+1}, y^t)
  \end{align*}
  where (a) uses~\eqref{eq:qb-def} and the induction hypothesis. Thus,
  \eqref{eq:pr-eq-2} holds for $t+1$ and, by the principle of induction,
  holds for all $t$. Hence~\eqref{eq:pr-eq} holds and, therefore, $I^{\mathbf
  q_a}(X_t, S_t; Y_t | Y^{t-1}) = I^{\mathbf q_b}(X_t, S_t; Y_t | Y^{t-1})$.
  The statement in the Lemma follows by adding over~$t$.
\end{proof}

%===============================================================
\section{Proof of \edit{concavity of the value function}} \label{app:concave}

To prove the result, we show the following:

\begin{lemsec} \label{lem:b1}
  For any action $a \in A$, if $V :\mathcal P_{X,S} \to \reals$ is concave, then $\mathscr B_a
  V$ is concave \edit{on $\mathcal P_{X,S}$}.
\end{lemsec}

The \edit{concavity of the value functions} follows from backward induction. $V_{T+1}$ is a
constant and, therefore, also concave. Lemma~\ref{lem:b1} implies that
$V_{T},V_{T-1}, \dots, V_1$ are concave.

\begin{proof} [Proof of Lemma~\ref{lem:b1}]
  The first term $I(a ; \pi)$ of $[\mathscr B_a V](\pi)$ is a concave function
  of $\pi$. We show the same for the second term.

  Note that if a function $V$ is concave, then it's perspective 
  ${g(u,t) := tV(u/t)}$ is concave in the domain 
  ${\{(u,t) : u/t \in \text{Dom}(V), t > 0\}}$. The second term in the
  definition of the Bellman operator~\eqref{eq:bellman} 
  \begin{align*}
    \sum_{y \in Y} \bigg[ \sum_{(x,s) \in \mathcal X \times \mathcal S} a(y|x,s) \pi(x,s) \bigg] V(\varphi(\pi,y,a)) 
  \end{align*}
  has this form because the numerator of $\varphi(\pi,y,a)$ is linear in $\pi$ and the denominator is $\sum_{x,s} a(y|x,s) \pi(x,s)$ (and corresponds to $t$ in the definition of perspective). Thus, for each $y$, the summand is concave in $\pi$, and the sum of concave functions is concave. Hence, the second term of the Bellman operator is concave in $\pi$. Thus we conclude that concavity is preserved under $\mathscr B_a$.
\end{proof}

%===============================================================
\section{Proof of Lemma \ref{lem:converge}} \label{app:converge}

The proof of the convergence of $\{\xi_t\}_{t\geq 1}$ relies on a result on
the convergence of partially observed Markov chains due to
Kaijser~\cite{Kaijser} that we restate below.

\begin{definition} A square matrix $D$ is called \emph{subrectangular} if for
  every pair of indices $(i_1, j_1)$ and $(i_2, j_2)$ such that $D_{i_1, j_1}
  \neq 0$ and $D_{i_2, j_2} \neq 0$, we have that $D_{i_2, j_1} \neq 0$ and
  $D_{i_1, j_2} \neq 0$.
\end{definition}

\begin{theorem}[Kaijser \cite{Kaijser}] 
  \label{thm:Kaijser} 
  Let $\left\{ U_t \right\}_{t \geq 1}$, $U_t \in \ALPHABET U$, be a finite
  state Markov chain  with transition matrix $P^u$. The initial state $U_1$ is
  distributed according to probability mass function $P_{U_1}$. Given a finite set
  $\ALPHABET Z$ and an observation function $g \colon \ALPHABET U \to
  \ALPHABET Z$, define the following:
  \begin{itemize}
    \item The process $\{Z_t\}_{t \ge 1}$, $Z_t \in \ALPHABET Z$,  given by
      \[  Z_t = g(U_t). \] 
    \item The process $\{\psi_t\}_{t \ge 1}$, $\psi_t \in \ALPHABET
      P_{\ALPHABET U}$, given by 
      \[ \psi_t(u) = \PR(U_t = u \mid Z^t). \]
    \item A square matrix $M(z)$, $z \in \ALPHABET Z$, given by
      \begin{equation*}
	\left[ M(z)\right]_{i,j} = 
	\begin{cases} 
	  P^u_{ij}	& \text{ if } g(j) = z\\
	  0       	& \text{ otherwise }
	\end{cases}	
	\qquad i,j \in \ALPHABET U.
      \end{equation*} 
  \end{itemize}

  If there exists a finite sequence $z^m_1$ such that $\prod_{t=1}^m M(z_t)$
  is subrectangular, then $\{\psi_t\}_{t \ge 1}$ converges in distribution to
  a limit that is independent of the initial distribution $P_{U_1}$.
\end{theorem}

We will use the above theorem to prove that under policy $\mathbf
b^{\circ}$, $\{\xi_t\}_{t \ge 1}$ converges to a limit. For that
matter, let $\ALPHABET U = \ALPHABET S \times \ALPHABET Y$, $\ALPHABET Z =
\ALPHABET Y$, $U_t = (S_t, Y_{t-1})$ and $g(S_t, Y_{t-1}) = Y_{t-1}$. 

First, we show that $\{U_t\}_{t \ge 1}$ is a Markov chain. In particular, for
any realization $(s^{t+1}, y^{t})$ of $(S^{t+1}, Y^t)$, we have that
\begin{align*}
    &\PR^{\mathbf b^\circ}
    (U_{t+1} = (s_{t+1},y_t) \mid U^{t} = (s^{t},y^{t-1})) \\
    &= \sum_{\tilde x_{t} \in \mathcal X} P(U_{t+1} = (s_{t+1},y_t), X_t =
    \tilde x_t \mid U^{t} = (s^{t},y^{t-1})) \\
    &= \sum_{\tilde x_{t} \in \mathcal X} \IND_{s_{t+1}}\{y_{t} + s_{t} - \tilde x_{t}\} b^*(y_t | s_t - \tilde x_t) P_X(\tilde x_{t})
    \\
    &=\PR^{\mathbf b^\circ}
    (U_{t+1} = (s_{t+1},y_t) \mid U_{t} = (s_{t},y_{t-1})).
\end{align*}

Next, let $m = 2m_s$ and consider 
\[
  z^m = 
  \underbrace{111\cdots1}_{m_s~\text{times}}
  \underbrace{000\cdots0}_{m_s~\text{times}}.
\]
We will show that this $z^m$ satisfies the subrectangularity condition of
Theorem~\ref{thm:Kaijser}. The basic idea is the following. Consider any
initial state $u_1 = (s, y)$ and any final state $u_m = (s',0)$. We will
show that 
\begin{align}
  \label{eq:to-smax}
  \PR(S_{m_s} = m_s \mid U_1 = (s,y), Z^{m_s} = (111\dots1)) &> 0, \\
  \intertext{and}
  \label{eq:from-smax}
  \PR(S_{2m_s} = s' \mid U_{m_s} = (s_m,1), Z_{m_s+1}^{m_s} = (000\dots0)) &>
  0.
\end{align}
Eqs.~\eqref{eq:to-smax} and~\eqref{eq:from-smax} show
that given the observation sequence $z^m$, for \emph{any} initial state
$(s,y)$ there is a positive probability of observing \emph{any} final state
$(s',0)$.\footnote{Note that given the observation sequence $z^m$, the final
state must be of the form $(s',0)$.} Hence, the matrix $\prod_{t=1}^m M(z)$ is
subrectangular. Consequently, by Theorem~\ref{thm:Kaijser}, the process
$\{\psi_t\}_{t \ge 1}$ converges in distribution to a limit that is
independent of the initial distribution $P_{U_1}$. 

Now observe that $\theta_t(s) = \sum_{y \in \ALPHABET Y} \psi_t(s,y)$ and
$(\theta_t, \xi_t)$ are related according to Lemma~\ref{lem:equiv}. Since
$\{\psi_t\}_{t \ge 1}$ converges weakly independent of the initial condition,
so do $\{\theta_t\}_{t \ge 1}$ and $\{\xi_t\}_{t \ge 1}$. 

Let $\bar \theta$ and $\bar \xi$ denote the limit of
$\{\theta_t\}_{t \ge 1}$ and $\{\xi_t\}_{t \ge 1}$. Suppose the initial
condition is $(\theta^\circ,\xi^\circ)$. Since $b^\circ$ is an invariant
policy, $(\theta_t, \xi_t) = (\theta^\circ, \xi^\circ)$ for all~$t$.
Therefore, the limits $(\bar \theta, \bar \xi) = (\theta^\circ, \xi^\circ)$. 

\begin{proof}[Proof of Eq.~\eqref{eq:to-smax}]
  Given the initial state $(s,y)$, define $\bar s = m_s - s$, and consider the
  sequence
  \[
    x^{m_s} = 
    \underbrace{000\cdots0}_{\bar s~\text{times}}
    \underbrace{111\cdots1}_{s~\text{times}}.
  \]
  Under this sequence of demands, consider the sequence of consumption
  $y^{m_s - 1} = (11\dots1)$, which is feasible because the state of the
  battery increases by~$1$ for the first $\bar s$ steps (at which time it
  reaches $m_s$) and then remains constant for the remaining $s$ steps.
  Therefore,
  \begin{multline*}
    \PR(S_{m_s} = m_s \mid U_1 = (s,y), \\
    Y^{m_s-1} = (111\dots1),
    X^{m_s} = x^{m_s}) > 0.
  \end{multline*}
  Since the sequnce of demands $x^m$ has a positive probability, 
  \begin{multline*}
    \PR(S_{m_s} = m_s, X^{m_s} = x^{m_s} \mid U_1 = (s,y), \\
    Y^{m_s-1} = (111\dots1)) > 0.
  \end{multline*}
  Therefore,
  \[
    \PR(S_{m_s} = m_s \mid U_1 = (s,y), Y^{m_s-1} = (111\dots1)) > 0
  \]
  which completes the proof.
\end{proof}

\begin{proof}[Proof of Eq.~\eqref{eq:from-smax}]
  The proof is similar to the Proof of~\eqref{eq:to-smax}. Given the final
  state $(s',0)$, define $\bar s' = m_s - s'$ and consider the sequence
  \[
    x_{m_s + 1}^{2m_s} =
    \underbrace{111\cdots1}_{\bar s'~\text{times}}
    \underbrace{000\cdots0}_{s'~\text{times}}.
  \]
  Under this sequnce of demains and the sequence of consumption given by
  $y_{m_s}^{2m_s - 1} = (00\dots0)$, the state of the battery decreases by~$1$
  for the first $\bar s'$ steps (at which time it reaches $s'$) and then
  remains constant for the remaining $s'$ steps. Since $x_{m_s + 1}^{2m_s}$
  has positive probability, we can complete the proof by following an argument
  similar to that in the proof of~\eqref{eq:to-smax}. 
\end{proof}

\bibliographystyle{IEEEtran}
\bibliography{IEEEabrv,sm}

\end{document}